%% file: main.tex
\documentclass[UKenglish]{lmcs}
\pdfoutput=1

\usepackage{lastpage}
\lmcsdoi{22}{3}{10}
\lmcsheading{}{\pageref{LastPage}}{}{}%
{Nov.~14,~2025}{Sep.~25,~2026}{}

\usepackage[utf8]{inputenc} % required by tex-lint
\usepackage{graphicx}
\usepackage{subcaption}
\usepackage{booktabs} % for tables without vertical rules

\usepackage[colorlinks=true,citecolor=blue,linkcolor=blue]{hyperref}
\usepackage{cleveref}

\usepackage{amsmath,amsthm,amssymb,amsfonts}
\usepackage{caption}
\usepackage{bbm}
\usepackage{multirow}
\usepackage{xspace}
\usepackage{float}

\usepackage{tikz}
\usetikzlibrary{arrows,shapes,automata, calc, chains, matrix, trees, positioning, scopes}
\usetikzlibrary{decorations.pathmorphing}
\tikzset{snake it/.style={decorate, decoration=snake}}

\newcommand{\OO}{\mathcal{O}}
\newcommand{\Aa}{\mathcal{A}}

\newcommand{\VV}{\mathcal{V}}
\newcommand{\Start}{\mathcal{S}}

\newcommand{\ACZ}{\textnormal{\textsf{AC$^0$}}\xspace}
\newcommand{\NCT}{\textnormal{\textsf{NC$^2$}}\xspace}

\newcommand{\NP}{\textnormal{\textsf{NP}}\xspace}
\newcommand{\NC}{\textnormal{\textsf{NC}}\xspace}

\newcommand{\PSPACE}{\textnormal{\textsf{PSPACE}}\xspace}

\newcommand{\NL}
{\textnormal{\textsf{NL}}\xspace}
\newcommand{\coNL}
{\textnormal{\textsf{coNL}}\xspace}

\newcommand{\PTIME}{\textnormal{\textsf{P}}\xspace}

\newcommand{\rkr}{\mathsf{rank}_{\mathbb{R}}}
\newcommand{\rku}{\mathsf{rank}_{\mathsf{un}}}
\newcommand{\rkg}{\mathsf{rank}}

\newcommand{\R}{\mathbb{R}}

\newcommand{\RP}{\mathbb{R}_{> 0}}

\newcommand{\spn}[1]{\mathord{\langle #1 \rangle}}
\newcommand{\Mer}{\mathsf{Mer}}

\newcommand{\mcw}{\mathsf{mcw}}
\newcommand{\mrw}{\mathsf{mrw}}
\newcommand{\MCol}{\mathsf{MCol}}
\newcommand{\MRow}{\mathsf{MRow}}

\tikzset{AUT style/.style={>=angle 60,every edge/.append style={},every state/.style={minimum size=16,inner sep=0}}}

\begin{document}

\title[The minimum rank of unambiguous finite automata]{Spectral and combinatorial methods for efficiently computing the rank of unambiguous finite automata}

\titlecomment{{\lsuper*}A preliminary version of this paper appeared at STACS 2025.}
\thanks{We thank the anonymous reviewers of the preliminary version for their helpful comments that improved the presentation of the paper. Andrew Ryzhikov is supported by Polish National Science Centre SONATA BIS-12 grant
number 2022/46/E/ST6/00230.}
% STACS 2025 special issue (exact title from episciences)
\specialissue{Selected Papers of the 42nd International Symposium on Theoretical Aspects of Computer Science (STACS 2025) - Track B}

\author[S.~Kiefer]{Stefan Kiefer\lmcsorcid{0000-0003-4173-6877}}[a]

\author[A.~Ryzhikov]{Andrew Ryzhikov\lmcsorcid{0000-0002-2031-2488}}[b]

\address{Department of Computer Science, University of Oxford, UK}
\email{stefan.kiefer@cs.ox.ac.uk}  

\address{University of Warsaw, Warsaw, Poland}	
\email{ryzhikov.andrew@gmail.com} 


\begin{abstract}
A zero-one matrix is a matrix with entries from $\{0, 1\}$. We study monoids containing only such matrices. A finite set of zero-one matrices generating such a monoid can be seen as the matrix representation of an unambiguous finite automaton, an important generalisation of deterministic finite automata which shares many of their good properties.

Let $\Aa$ be a finite set of $n \times n$ zero-one matrices generating a monoid of zero-one matrices, and $m$ be the cardinality of $\Aa$.
We study the computational complexity of computing the minimum rank of a matrix in the monoid generated by $\Aa$. By using linear-algebraic techniques, we show that this problem is in~\NC and can be solved in $\OO(mn^4)$ time and $\OO(n^2)$ space. We also provide a combinatorial algorithm finding a matrix of minimum rank in $\OO(mn^4)$ time and $\OO(n^3)$ space. As a byproduct, we show a very weak version of a generalisation of the \v{C}ern\'{y} conjecture: there always exists a straight line program of size $\OO(n^2)$ describing a product resulting in a matrix of minimum rank.

For the special case corresponding to total DFAs (that is, for the case where all matrices have exactly one 1 in each row), the minimum rank is the size of the smallest image of the set of all states under the action of a word. Our combinatorial algorithm finds a matrix of minimum rank in time $\OO(n^3 + mn^2)$ in this case.
\end{abstract}

\maketitle

\section{Introduction}\label{sec-intro}
\input{sec-intro-new}

\section{Existing results and our contributions}\label{sec-existing}
\input{sec-existing}

\section{Main definitions}\label{sec-definitions}
\input{sec-defs-new}

\section{Main concepts and the linear algebra toolbox}\label{sec-toolbox}
\input{sec-toolbox}

\section{Computing the rank in \texorpdfstring{\NCT}{NC2}}\label{sec-nc2}
\input{sec-nc2-new}

\section{Time and space complexity}\label{sec-time}
\input{sec-time}

\section{Algebraic synchronisation criterion}\label{sec-algebraic}
\input{sec-algebraic-criterion}

\section{Conclusions and open problems}\label{sec-conclusions}
\input{sec-conclusions}

\bibliographystyle{alphaurl}
\bibliography{sample}

\end{document}

%% file: sec-intro-new.tex
Matrix monoids are a rich and versatile object naturally appearing in formal verification, program analysis, dynamical systems and weighted automata. 
However, many of their properties are in general undecidable. One such example is the well-studied matrix mortality problem. Given a finite set~$\Aa$ of $n \times n$ matrices, it asks if the monoid generated by $\Aa$ (that is, the set of all products of matrices from $\Aa$) contains the zero matrix. This problem is undecidable already for $3 \times 3$ integer matrices~\cite{Paterson1970}, and was studied for several decidable special cases, see e.g.~\cite{Cassaigne2014,Bell2021,Ryzhikov2024RP}.

Even if $\Aa$ is a set of two zero-one matrices (that is, matrices with entries in~$\{0, 1\}$), matrix mortality is \PSPACE-complete~\cite{Ryzhikov2024RP}. Thus, to make it tractable, one has to further restrict the problem. In this paper, we consider the case where the whole monoid generated by $\Aa$ consists of zero-one matrices;  matrix mortality then becomes decidable in polynomial time~\cite{Kiefer2021}. We call such monoids \emph{zero-one matrix monoids}. Intuitively, when multiplying any two matrices from such a monoid, we never get $1 + 1$ as a subexpression. Zero-one matrix monoids have a rich structure while still admitting good algorithmic properties. They correspond precisely to unambiguous finite automata, and find applications in formal verification \cite{Baier2023}, variable-length codes \cite{Berstel2010} and symbolic dynamics \cite{Lind2021}.
They are also an interesting special case of finite monoids of rational matrices (studied in, e.g., \cite{Mandel1977,Jacob1977,Almeida2009,Bumpus2020}), monoids of nonnegative matrices (studied in, e.g., \cite{Protasov2021,Blondel2015,Wu2023,Gerencser2018}), and, in the case where they do not contain the zero matrix, of matrix monoids with constant spectral radius~\cite{Protasov2017}.

In this paper, we consider a problem that can be seen as a natural generalisation of matrix mortality: given a finite set $\Aa$ generating a zero-one matrix monoid, find the minimum real rank of a matrix in this monoid. By the real rank of a matrix we mean the dimension of the subspace generated by its columns over the reals. Clearly, this rank is zero if and only if the monoid contains the zero matrix. 
The minimum real rank of a matrix in a zero-one matrix monoid is a much more tractable problem than deciding other similar properties: for example,  checking if a zero-one matrix monoid contains a matrix of a given real rank was shown to be \NP-hard\footnote{In fact, it is \PSPACE-complete, which follows directly from \cite[Theorem~3]{Berlinkov2014}: add a fresh state and define all yet undefined transitions to lead to this state.} \cite{GORALCIK1995}, and checking if it contains a given matrix is a classical \PSPACE-complete problem~\cite{Kozen1977}.

An important motivation for considering the minimum rank in a zero-one matrix monoid comes from a probabilistic perspective. Let $\Aa$ be a finite set of matrices, and consider a product of $\ell$ matrices where at each position a matrix from $\Aa$ is chosen uniformly at random. It is easy to see that when $\ell$ tends to infinity, the probability that the rank of this product is equal to the minimum rank of matrices from the monoid generated by $\Aa$ converges to one. Thus, this minimum rank characterises the most likely eventual behaviour of a linear dynamical system corresponding to $\Aa$.

The goals of our paper are as follows.
\begin{itemize}
    \item We present efficient algorithms for analysing monoids of zero-one matrices and unambiguous finite automata.

    \item To obtain these algorithms, we provide new structural and algebraic properties of such monoids and automata that might be interesting on their own.

    \item In particular, we provide an algebraic obstacle for an unambiguous finite automaton to be synchronising, similar to the known result for deterministic finite automata. 

    \item We thus strengthen the connections between the areas of synchronising automata, weighted automata and matrix semigroups by transferring methods and tools between them.

    \item Finally, we highlight open problems in the intersection of these areas.
\end{itemize}

%% file: sec-existing.tex
Throughout the paper, we always assume that matrix monoids are defined by sets of generators, and 
all matrices are square zero-one unless stated otherwise.

\subsection{Total DFAs}

An $n \times n$ zero-one matrix with exactly one $1$ in every row can be equivalently seen as a transformation of a set $Q$ of size $n$. A set of such matrices generates a zero-one matrix monoid, and can be seen as a total deterministic finite (semi-)automaton\footnote{\label{footnote-semi}In this paper, all automata are semi-automata, meaning that they do not have any initial or accepting states, and thus do not recognise any languages. Following the usual conventions (as in, e.g.,~\cite{Berstel2010}), we omit ``semi-'', in particular because it would make abbreviations like DFA less recognisable. Total DFAs are often called complete, but we prefer the
term “total” both to avoid the clash of terminology with complete UFAs and because of the obvious connection
with total functions.} (total DFA) $\Aa = (Q, \Sigma, \delta)$. Here, $\Sigma$ is a finite alphabet whose letters correspond to the generating matrices, and $\delta: Q \times \Sigma \to Q$ is the transition function defined in such a way that for each $a \in \Sigma$, $\delta(\underline{\hspace{0.25cm}}, a)$ is the transformation of $Q$ induced in the natural way by the matrix corresponding to $a$. Thus, words over $\Sigma$ correspond to products of the generating matrices.

The \emph{rank} of a word $w$ in $\Aa$ is the size of the image of $Q$ under the transformation corresponding to $w$. Equivalently, it is the real rank of the matrix corresponding to $w$.
The \emph{rank} of a total DFA is the minimum among the ranks of all its words. This concept was studied from the perspectives of automata theory \cite{Rystsov1992Rank,Kari2019} and the theory of transformation semigroups \cite{Shin2010Rank,Karimi2017MinIdeal}. It is the subject of the rank conjecture (called the \v{C}ern\'{y}-Pin conjecture in \cite{Rystsov1992Rank}), which states that every total DFA of rank $r$ admits a word of rank $r$ having length at most $(n - r)^2$. The \v{C}ern\'{y} conjecture, one of the oldest open problems in combinatorial automata theory \cite{Volkov2022}, is a special case with $r = 1$. We refer to surveys \cite{Volkov2008,Beal2016,Kari2019,Volkov2022} for the vast literature on the \v{C}ern\'{y} conjecture. Underlying digraphs of total DFAs of a given rank were studied in \cite{Budzban2011Generalized,Beal2016} in the context of the road colouring problem.

The rank of an $n$-state total DFA over an alphabet of size $m$ can be found in $\OO(m^4n^4)$ time \cite[Theorem 1]{Rystsov1992Rank}.
In contrast, for any fixed $r \ge 2$, the problem of checking if a total DFA admits a word of rank $r$ is \NP-hard~\cite{GORALCIK1995}. Checking if an $n$-state total DFA over an alphabet of size $m$ has rank one is \NL-complete \cite{Holzer2018,Volkov2022}, and can be done in $\OO(m n^2)$ time \cite{Eppstein1990,Volkov2008}. For each total DFA of rank $r$, there exists a word of rank $r$ of length at most $\frac{(n - r)^3}{6} + \OO((n - r)^2)$~\cite{Klyachko1987}, and if $r = 1$,  finding a word of rank one can be done in $\OO(n^3 + mn^2)$ time and $\OO(n^2)$ space~\cite{Eppstein1990}.

\subsection{Unambiguous finite automata}

Generalising the case of total DFAs, a set $\Aa$ of $n \times n$ zero-one matrices generating a zero-one matrix monoid can be equivalently seen as an unambiguous nondeterministic finite (semi-)automaton (UFA). Let $Q = \{q_1, \ldots, q_n\}$ be its set of states. To each matrix in $\Aa$ we again associate a letter in the alphabet $\Sigma$, and the transition relation $\Delta \subseteq Q \times \Sigma \times Q$ is defined so that $(q_i,a,q_j) \in \Delta$ if and only if the entry~$(i, j)$ in the matrix corresponding to $a$ is equal to one. Just as in the total DFA case, words over $\Sigma$ naturally correspond to products of matrices from $\Aa$.

The obtained NFA then has the property that is sometimes called \emph{diamond-free}: for every two states $p, q$ and every word~$w$, there is at most one path from $p$ to $q$ labelled by $w$. A simple reachability argument shows that the length of a shortest word labelling two such paths, if it exists, is at most quadratic in the dimension of the matrices. Hence, deciding whether an NFA is a UFA (and thus whether a set of zero-one matrices generates a zero-one monoid)  is in \coNL = \NL. It is actually \NL-complete as described in the next subsection.

A UFA is called complete if it does not admit a word whose matrix is the zero matrix. For an $n$-state UFA the length of such a word if it exists is at most~$n^5$~\cite{Kiefer2021}. 
The best known lower bound is quadratic in $n$, and is achieved by a series of DFAs~\cite{Rystsov1997}. 
For UFAs, the quadratic upper bound was conjectured to be tight~\cite[Conjecture 2]{Rystsov1992Rank}. Checking if a UFA is complete can be done in \NCT~\cite{Kiefer2021}. 

The \emph{real rank} of a UFA is the minimum among the real ranks of the matrices corresponding to words. It was shown in \cite{Carpi1988} that for an $n$-state UFA of real rank $r \ge 1$ there always exists a word of minimum rank of length $\OO(rn^3)$.
For $n$-state strongly connected Eulerian UFAs of rank one, a subclass with remarkably nice properties, there always exists a word of length at most $(n - 1)^2$ of rank one~\cite[Corollary~4]{Carpi2009}. All mentioned constructions also provide polynomial time algorithms that construct words with the required properties (in particular, with a length within the stated bounds).

\subsection{Applications to variable-length codes}\label{subs:var-length-codes}

A \emph{variable-length code} (or simply a \emph{code}) is a set $X$ of finite words over an alphabet $\Sigma$ such that every finite word over $\Sigma$ has at most one factorisation over $X$. In other words, a code is a basis of a free submonoid of $\Sigma^*$. 

The definitions of both UFAs and codes rely, intuitively, on the uniqueness of certain representations. In fact, UFAs and codes are tightly related. Let us illustrate this relationship.
If the cardinality of a code $X$ is finite, one can construct its flower automaton, which is a UFA with a chosen state $s$ such that, for each word from $X$, there is a separate cycle containing $s$ and labelled by this word, see \autoref{fig:flower} (left) for an example.
More generally, codes that are regular languages correspond precisely to strongly connected UFAs in a similar way, see \cite[Chapter 4]{Berstel2010} for the details.

\begin{figure}[h]
\centering
% Center the two drawings as a group
\begin{tabular}{@{}c@{\hspace{1cm}}c@{}}
\begin{tikzpicture} [node distance = 2cm]
%,
\tikzset{every state/.style={inner sep=1pt,minimum size=1.5em}}

\node [state] at (0, 0) (1) {$1$};
\node [state] at (-2, 0) (2) {$2$};
\node [state] at (-1, 1.5) (3) {$3$};
\node [state] at (1, 1.5) (4) {$4$};

\node [state] at (1.5, 0.75) (5) {$5$};
\node [state] at (1.5, -0.75) (6) {$6$};

\node [state] at (1, -1.5) (7) {$7$};
\node [state] at (0, -1.7) (8) {$8$};
\node [state] at (-1, -1.5) (9) {$9$};

\path [-stealth, thick]

(1) edge [bend left=30] node[above] {$a$} (2)
(2) edge [bend left=30] node[below] {$a$} (1)

(1) edge [bend left=10] node[right] {$a$} (3)
(3) edge [bend left=10] node[below] {$a$} (4)
(4) edge [bend left=10] node[left] {$b$} (1)

(1) edge [bend left=10] node[below] {$a$} (5)
(5) edge [bend left=10] node[left] {$b$} (6)
(6) edge [bend left=10] node[above] {$a$} (1)

(1) edge [bend left=10] node[left] {$a$} (7)
(7) edge [bend left=10] node[above] {$b$} (8)
(8) edge [bend left=10] node[above] {$a$} (9)
(9) edge [bend left=10] node[right] {$b$} (1)

;
\end{tikzpicture}
&
\begin{tikzpicture} [node distance = 2cm]

\node [] at (0.5, 2) (t0) {\strut $\cdots$};
\node [] at (0.5, 2.375) (t0u) {};
\node [] at (0.5, 1.625) (t0l) {};
\node [] at (1, 2) (t1) {\strut $a$};
\node [] at (1.5, 2) (t2) {\strut $b$};
\node [] at (1.75, 2) (t23) {};
\node [] at (2, 2) (t3) {\strut $a$};
\node [] at (2.5, 2) (t4) {\strut $b$};
\node [] at (3, 2) (t5) {\strut $a$};
\node [] at (3.25, 2) (t56) {};
\node [] at (3.75, 2.375) (t6u) {};
\node [] at (3.75, 1.625) (t6l) {};
\node [] at (3.5, 2) (0) {\strut $\cdots$};

%%%%%%%%%%%%%%%%%%%%%%
\node [] at (-0.5, 0) (0) {\strut $\cdots$};
\node [] at (-0.5, -0.375) (0l) {};
\node [] at (-0.5, 0.25) (0) {};
\node [] at (0, 0) (1) {\strut $b$};
\node [] at (0.25, 0) (12) {};
\node [] at (0.5, 0) (2) {\strut $a$};
\node [] at (0.75, 0) (23) {};
\node [] at (1, 0) (3) {\strut $b$};
\node [] at (1.25, 0) (34) {};
\node [] at (1.5, 0) (4) {\strut $a$};
\node [] at (1.75, 0) (45) {};
\node [] at (2, 0) (5) {\strut $b$};
\node [] at (2.25, 0) (56) {};
\node [] at (2.5, 0) (6) {\strut $a$};
\node [] at (2.75, 0) (67) {};
\node [] at (3, 0) (7) {\strut $a$};
\node [] at (3.25, 0) (78) {};
\node [] at (3.5, 0) (8) {\strut $b$};
\node [] at (3.75, 0) (89) {};
\node [] at (4, 0) (9) {\strut $a$};
\node [] at (4.75, 0.25) (10) {};
\node [] at (4.25, 0) (10l) {};
\node [] at (4.75, -0.25) (10ll) {};
\node [] at (4.5, 0) (11) {\strut $\cdots$};

\path [thick]

(t0u) edge [bend left=30] node[above] {} (t23)
(t23) edge [bend left=40] node[above] {} (t56)
(t56) edge [bend left=40] node[above] {} (t6u)

(t0l) edge [bend right=30] node[above] {} (t23)
(t23) edge [bend right=30] node[above] {} (t6l)

%%%%%%%%%%%%%%%%%%

(0) edge [bend left=30] node[above] {} (12)
(12) edge [bend left=40] node[above] {} (56)
(56) edge [bend left=50] node[above] {} (89)
(89) edge [bend left=30] node[above] {} (10)

(0l) edge [bend right=30] node[above] {} (34)
(34) edge [bend right=50] node[above] {} (67)
(67) edge [bend right=50] node[above] {} (10l)
(10l) edge [bend right=40] node[above] {} (10ll)

;
\end{tikzpicture}
\end{tabular}
\caption{The flower automaton of the code $X = \{aa, aab, aba, abab\}$ (left), two adjacent interpretations of $ababa$ over $X$ (top right), and two disjoint interpretations of $bababaaba$ over $X$ (bottom right). Note that this code is not complete, but still illustrates all the discussed properties.}\label{fig:flower}
\end{figure}
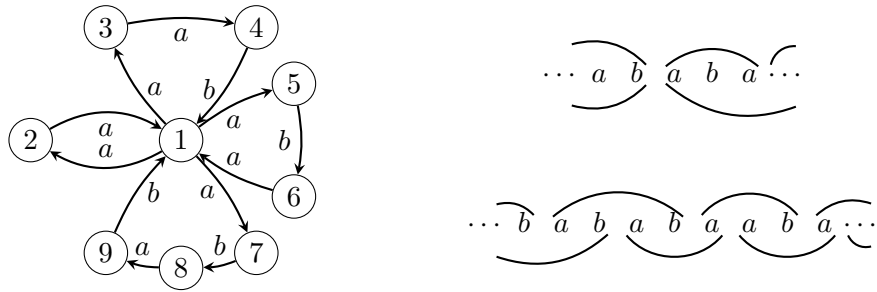

% Fix indentation
\noindent
A useful application is the fact that deciding if a set of zero-one matrices generates a zero-one monoid is \NL-hard. Indeed, a finite set of words is a code if and only if its flower automaton is unambiguous \cite{Berstel2010}.  
Deciding if a finite set of words is a code is \NL-complete \cite{Rytter1986}, and the flower automaton can be constructed in~\ACZ.

A code $X$ over $\Sigma$ is called \emph{complete} if every word over $\Sigma$ is a factor of a concatenation of codewords, that is, for every word $w \in \Sigma^*$ there exist $u, v \in \Sigma^*$ with $uwv \in X^*$.
A code that is a regular language is complete if and only if the corresponding UFA is complete~\cite{Berstel2010}.  
For complete codes that are regular languages, the real rank of the corresponding UFA is equal to a natural and  important parameter called the degree of a code~\cite[Proposition~9.6.1]{Berstel2010}.

\subsection{The degree of a code}
Let us first explain the intuition behind the notion of degree. For each word $w$ we can consider all possible factorisations over $X$ of all its extensions $uwv \in X^*$ with $u, v \in \Sigma^*$, called \emph{interpretations} of~$w$. Two such interpretations either match in at least one position (as in \autoref{fig:flower} (top right) between the second and the third letter), or do not match in any position (as in \autoref{fig:flower} (bottom right)), in which case they are called \emph{disjoint}. The degree of a word is the number of pairwise disjoint interpretations of this word. The degree of a code
$X$ is the minimum nonzero degree of all words $w \in \Sigma^*$. 

Formally, an \emph{interpretation} of a word $w$ over a code $X$ is a triple $(d, x, g)$ such that $d$ is a suffix of a word from $X$, $x \in X^*$, and $g$ is a prefix of a word from $X$. Two interpretations $(d, x, g)$ and $(d', x', g')$ of $w$ are said to be \emph{adjacent} if there exist $y, z, y', z' \in X^*$ with 
$x = yz, x' = y'z', dy = d'y', zg = z'g'.$ Two interpretations are said to be \emph{disjoint} if they are not adjacent. See \autoref{fig:flower} for an example. The \emph{degree} of a word is the number of pairwise disjoint interpretations of this word. The \emph{degree} of a code $X$ is the minimum nonzero degree of all words $w \in \Sigma^*$. For codes that are regular languages, the degree is equal to the minimum nonzero rank of the corresponding UFA \cite[Proposition 9.6.1]{Berstel2010}. In particular, if there exists a word $w \in X^*$ of degree $1$ (called a synchronising word) for a code $X$, then any sequence $xwwy \in X^*$ of codewords with $x, y \in X^*$ is guaranteed to split into $xw, wy \in X^*$, thus allowing independent decoding of the two halves.

A particularly important case is when a complete code has degree one. Then
there exists a word $w \in X^*$ (called a synchronising word) such that for any concatenation of codewords $uwwv \in X^*$  with $u, v \in \Sigma^*$ we have $uw, wv \in X^*$. Intuitively, this means that the two halves $uw$ and $wv$ can be decoded separately and independently.

\subsection{Computational complexity classes} 
In this paper, we characterise the computational complexity of problems by showing that they belong to the classes $\NL \subseteq \NCT \subseteq \NC \subseteq \PTIME$, see \cite{Arora2009, Goldreich2008} for their formal definitions. \NL~is the class of problems solvable in nondeterministic logarithmic time. $\NC^k$ is the class of problems solvable by $\OO((\log n)^k)$-depth polynomial-size bounded fan-in Boolean circuits, and \NC is the union of these classes for all $k \ge 1$.
The class \NC represents problems that have efficient parallel algorithms, and is a subclass of problems solvable in polylogarithmic space~\cite{Arora2009}. Intuitively, \NC is the class of problems that can be solved using local computations, as opposed to \PTIME-complete problems, which are inherently sequential and thus require storing the entire structure in the memory unless $\NC = \PTIME$. An important property of \NC is that problems from this class can be used for designing \PSPACE algorithms as discussed, e.g., in the beginning of ~\cite[Section 3]{Jain2011}. Namely, the  composition of a \PSPACE-transducer and an \NC-algorithm is a \PSPACE-algorithm~\cite{Borodin1977}, despite the fact that the output of a \PSPACE-transduction can have exponential size. This is not necessarily true when an \NC-algorithm is replaced by an arbitrary polynomial time algorithm. In the context of formal verification, this compositional approach is used, e.g., in~\cite{Baier2023}.

% Keep the last lines of 2.5 on this page (avoid a stretched gap before 2.6)
\enlargethispage{2\baselineskip}
\NCT is an especially important class in computational algebra.
To quote \cite[page 468]{Goldreich2008}, ``\NCT is the habitat of most natural problems in linear algebra''. Indeed, matrix multiplication, computing the determinant, inverse and rank of a matrix belong to \NCT~\cite{BorodinGH82,Cook85,Berkowitz1984,Csanky1976}.

\subsection{Our contributions}
The known results about reachability properties of zero-one matrix monoids (including the special case of total DFAs), such as \cite{Carpi1988,Eppstein1990,Ryzhikov2019WORDS,Kiefer2021}, mostly construct a product of minimum rank iteratively, with each iteration decreasing the number of different rows or the rank of a matrix.
Such an approach is inherently sequential, since the matrix in the new iteration has to depend on the previous one, which thus has to be constructed explicitly.
In particular, this requires matrix multiplication at every step, which heavily increases the time complexity. In this paper, we take a different direction by strongly relying on linear algebra. While linear-algebraic arguments are used widely in the synchronising automata literature, they mostly serve to decrease the number of iterations in the iterative scheme described above. Our approach is to instead relate the rank of a zero-one matrix monoid to efficiently computable linear-algebraic properties, without explicitly constructing a matrix of minimum rank.

Our first main result is that computing the rank of a zero-one matrix monoid provided in the input by a generating set of $m$ matrices of dimension $n$ (or, equivalently, by a UFA with $n$ states and $m$ letters) is in \NCT (\autoref{thm:rank-in-NC}) and can be done in time $\OO(mn^4)$ (\autoref{thm:linalg-time}) and space $\OO(n^2)$. Previously, it was not known that this problem is in \NC, not even for total DFAs or finite complete codes. Moreover, the naive implementation of the polynomial time algorithm from the literature works in time~$\OO(n^{4 + \omega} + mn^4)$~\cite{Carpi1988}.

Our results rely on a new concept of weight of the matrices in a complete zero-one monoid.
This theory of matrix weight, which we develop in \cref{sec-toolbox}, is our main technical contribution.
Matrix weight is a natural generalisation of an existing notion of weight of columns of matrices in total DFAs, which was used, e.g., in connection with the road colouring problem~\cite{Friedman1990,Kari2001,Gusev2016}.
We show that all matrices in a zero-one matrix monoid have the same weight, and that this weight is tightly related to both the rank of the monoid and to the maximal weight of the columns and rows of its matrices (\autoref{subsec:weight-preserv}).
This connection allows us to reduce the computation of the monoid rank to the computation of maximal column and row weight.
Then we show that we can instead compute the weight of ``maximal pseudo-columns'' and ``maximal pseudo-rows'', as they have the same weight as maximal columns and rows, respectively (\autoref{subsec:pseudocolumns}).
Finally, we transfer linear-algebraic techniques from the literature on weighted automata to compute those weights, and thus the rank of the monoid, efficiently (\autoref{sec-nc2} and \autoref{subsec:linalg-time}).

We complement the linear-algebraic algorithms with a combinatorial algorithm, our second main contribution. While it has the same time complexity of  $\OO(mn^4)$ and a higher space complexity of $\OO(n^3)$ in the general case (\autoref{thm:finding-matrix}), it also constructs a matrix of minimum rank in addition to computing the rank of the monoid. For total DFAs, our combinatorial algorithm runs in time $\OO(n^3 + mn^2)$ (\autoref{thm:main-dfas}), thus outmatching the linear-algebraic counterpart and improving upon the $\OO(m^4n^4)$ algorithm known before~\cite{Rystsov1992Rank}. The key technical ingredients of our combinatorial algorithm are explained in the beginnings of \autoref{subsec:max-col} and \autoref{subsec-finding-mat}.
Our results on the time and space complexity of computing the rank are summarised in the table below, in the format ``time complexity, space complexity, reference for both''.

\begin{center}
% booktabs: no vertical rules
\begin{tabular}{ccc}
  \toprule
  class & UFA & total DFA \\
  \midrule
  previous best & $\OO(n^{4 + \omega} + mn^4)$, $\OO(n^4)$~\cite{Carpi1988} & $\OO(m^4n^4)$, $\OO(n^4)$~\cite{Rystsov1992Rank} \\
  linear-algebraic & $\OO(mn^4)$,  $\OO(n^2)$~(\autoref{thm:linalg-time}) & $\OO(mn^3)$, $\OO(n^2)$~(see \autoref{subsec:linalg-time}) \\
  combinatorial & $\OO( mn^4)$, $\OO(n^3)$ (\autoref{thm:finding-matrix}) & $\OO(n^3 + mn^2)$, $\OO(n^2)$~(\autoref{thm:main-dfas}) \\
  \bottomrule
\end{tabular}
\end{center}

%% file: sec-defs-new.tex
Let $Q$ be a finite set, which we view as a set of states.
For $S \subseteq Q$ we write $[S]$ for the column vector $x \in \{0,1\}^Q$ such that $x(q) = 1$ if and only if $q \in S$.
We may write $[q]$ for~$[\{q\}]$.
For a column vector $x \in \{0,1\}^Q$ we write $x^T$ for the transpose, a row vector.
For two column vectors $x_1,x_2 \in \mathbb{R}^Q$ we write $x_1 \ge x_2$ if the inequality holds component-wise.
We view the elements of $\mathbb{R}^{Q \times Q}$ (and similar sets) as matrices.
Vector and matrix addition and multiplication are defined in the usual way (over $\mathbb{R}$).
We denote by $\spn{X}$ the span of a set~$X$ of vectors, i.e., the set of all linear combinations of~$X$ with real coefficients.
The \emph{real rank} of a matrix $A \in \mathbb{R}^{Q \times Q}$ is, as usual, the dimension of the column space of~$A$ over the field of the reals (which equals the dimension of the row space); i.e., $\rkr(A) = \dim \spn{A [q] \mid q \in Q} = \dim \spn{[q]^T A \mid q \in Q}$.

Let $\Aa = \{A_1, \ldots, A_m\}$ be a set of matrices from $\{0,1\}^{Q \times Q}$, and $\Sigma = \{a_1, \ldots, a_m\}$ be a finite alphabet.
We associate the letters with the matrices by setting $M(a_i) = A_i$ for $1 \le i \le m$.
Throughout this paper, when speaking about computational complexity, we assume that the input is the function $M \colon \Sigma \to \{0,1\}^{Q \times Q}$ from letters to zero-one matrices.
We can extend  $M \colon \Sigma \to \{0,1\}^{Q \times Q}$ naturally (and often implicitly) to $M \colon \Sigma^* \to \mathbb{Z}_{\ge 0}^{Q \times Q}$ by defining $M(a_1 \cdots a_k) = M(a_1) \cdots M(a_k)$. %, where the multiplication operation on the right-hand side is usual matrix multiplication.
Thus, $M$~is a monoid homomorphism from~$\Sigma^*$ to the matrix monoid $M(\Sigma^*)$ generated by $\Aa = M(\Sigma)$.
Note that $M(\varepsilon) = I$, where $\varepsilon$~denotes the empty word and $I$ the $Q \times Q$ identity matrix. In this paper, we consider only monoid morphisms $M \colon \Sigma^* \to \mathbb{Z}_{\ge 0}^{Q \times Q}$ that are \emph{unambiguous}, i.e., $M \colon \Sigma^* \to \{0,1\}^{Q \times Q}$. If $M$ is unambiguous, $\Aa = M(\Sigma)$ generates a finite matrix monoid $M(\Sigma^*) \subseteq \{0,1\}^{Q \times Q}$.

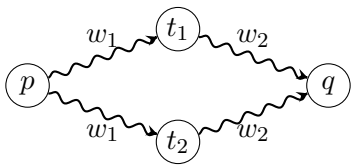
\begin{figure}[h]\centering
\begin{tikzpicture} [node distance = 2cm]
%,
\tikzset{every state/.style={inner sep=1pt,minimum size=1.5em}}

\node [state] at (-2, 0) (p) {$p$};
\node [state] at (0, 0.75) (t1) {$t_1$};
\node [state] at (2, 0) (q) {$q$};
\node [state] at (0, -0.75) (t2) {$t_2$};
\path [-stealth, thick]

(p) edge [decorate, decoration={snake, segment length=3mm, amplitude=0.5mm}] node[above] {$w_1$} (t1)
(t1) edge [decorate, decoration={snake, segment length=3mm, amplitude=0.5mm}] node[above] {$w_2$} (q)

(p) edge [decorate, decoration={snake, segment length=3mm, amplitude=0.5mm}] node[below] {$w_1$} (t2)
(t2) edge [decorate, decoration={snake, segment length=3mm, amplitude=0.5mm}] node[below] {$w_2$} (q)
;
\end{tikzpicture}
%\raisebox{-0.5\height}{
\caption{The configuration that is forbidden in a UFA.}\label{fig:diamond-free}
\end{figure}

% Fix indentation
\noindent
Viewing the matrices as transition matrices of an automaton, we obtain a \emph{nondeterministic finite (semi-)automaton (NFA)} $(Q, \Sigma, \Delta)$ with transition relation $\Delta = \{(p,a,q) \in Q \times \Sigma \times Q \mid [p]^T M(a) [q] = 1\}$.
Recall that in this paper automata do not have dedicated initial or accepting states, see footnote~\ref{footnote-semi} on page~\pageref{footnote-semi}.
We can extend $\Delta$ from letters to words in the usual way so that we have $\Delta = \{(p,w,q) \in Q \times \Sigma^* \times Q \mid [p]^T M(w) [q] \ge 1\}$.
An NFA $(Q, \Sigma, \Delta)$ is \emph{unambiguous}\footnote{In the context of finite automata that recognise languages, the usual notion of unambiguity also depends on the choice of initial and final states, and is thus not reflected in the transition monoid. Our definition of unambiguity thus defines a strictly larger class of NFAs. Its advantage is that it is a property of the transition monoid alone, in the same way as determinism.} (or \emph{diamond-free}) if for every two states $p,q$ and for every two words $w_1, w_2$ there exists at most one $t \in Q$ with $(p,w_1,t) \in \Delta$ and $(t,w_2,q) \in \Delta$; see \autoref{fig:diamond-free} for an illustration of the forbidden configuration. We denote unambiguous NFAs as UFAs. Recall from the previous section that deciding if an NFA is unambiguous is \NL-complete. In the following, we often identify $M \colon \Sigma^* \to \{0,1\}^{Q \times Q}$ with the corresponding UFA $(Q, \Sigma, \Delta)$. In particular, a monoid homomorphism is unambiguous if and only if the corresponding NFA is unambiguous.

When $M$ (or, equivalently, $\Delta$) is clear from the context, we may write $p \cdot w = \{q \in Q \mid (p,w,q) \in \Delta\}$.
Then $[p \cdot w]^T = [p]^T M(w)$.
Similarly, we may write $w \cdot q = \{p \in Q \mid (p,w,q) \in \Delta\}$, so that $[w \cdot q] = M(w) [q]$. 
We call $M$ \emph{strongly connected} if for all $p,q \in Q$ there is $w \in \Sigma^*$ with $p \cdot w \ni q$. 
We call $M$ \emph{complete} if $0 \not\in M(\Sigma^*)$, where $0$ is the zero matrix.
The \emph{real rank} of~$M$ (and of $M(\Sigma^*)$) is 
\[\rkr(M) := \min\{\rkr(M(w)) \mid w \in \Sigma^*\}.\]
Note that~$M$ is complete if and only if $\rkr(M) \ne 0$.

Suppose that $|p \cdot a| = 1$ holds for every $p \in Q$ and $a \in \Sigma$, or, equivalently, that every matrix in~$\Aa$ has exactly one~$1$ in each row.
Then $|p \cdot w| = 1$ holds for every $p \in Q$ and $w \in \Sigma^*$.
We call such UFAs \emph{total deterministic finite (semi-)automata (total DFAs)} and we may write $\delta$ instead of $\Delta$ to highlight that it is a transition function $\delta \colon Q \times \Sigma \to Q$ instead of a transition relation.
A total DFA $(Q, \Sigma, \delta)$ is complete in the sense defined above (i.e., $0 \not\in M(\Sigma^*)$), and for any $w \in \Sigma^*$ we have that $\rkr(M(w))$ is the number of nonzero columns in~$M(w)$.

%% file: sec-toolbox.tex
In this section, we introduce the main tools that we will use for both linear-algebraic and combinatorial algorithms in later sections.
Until \autoref{subsec:non-sc}, we fix an unambiguous, complete, and strongly connected monoid morphism~$M$. In \autoref{subsec:non-sc} we will show that the case where $M$ is not strongly connected can be easily reduced to the strongly connected case.

\subsection{Columns, rows and the structure of minimum rank matrices}\label{subsec:col-row-defs}

The concept of maximum columns and rows plays a crucial role in dealing with reachability problems in unambiguous monoid morphisms. Abusing language slightly in the following, by \emph{column} we refer to column vectors of the form $[w \cdot q] = M(w)[q] \in \{0,1\}^Q$ where $w \in \Sigma^*$ and $q \in Q$.
Similarly, a \emph{row} is of the form $[q \cdot w]^T = [q]^T M(w)$. See \autoref{fig:ex-columns-rows} for an example. In the case of total DFAs, all rows are of the form $[q]^T$. This fact makes total DFAs significantly simpler to deal with than general complete UFAs.

\begin{figure}[h]
\centering
% Center the two drawings as a group
\begin{tabular}{@{}c@{\hspace{1cm}}c@{}}
$
M(a) = \begin{pmatrix} 1 & 0 & 1 & 0 \\ 1 & 0 & 1 & 0 \\ 0 & 0 & 0 & 0 \\ 0 & 0 & 0 & 0 \end{pmatrix} \quad 
M(b) = \begin{pmatrix} 0 & 0 & 0 & 0 \\ 0 & 0 & 0 & 0 \\ 0 & 1 & 0 & 1 \\ 0 & 1 & 0 & 1 \end{pmatrix}
$
&
\begin{tikzpicture} [baseline=(current bounding box.center), node distance = 2cm]
%,
\tikzset{every state/.style={inner sep=1pt,minimum size=1.5em}}

\node [state] at (-2, 0) (1) {$1$};
\node [state] at (0, 1) (2) {$2$};
\node [state] at (2, 0) (3) {$3$};
\node [state] at (0, -1) (4) {$4$};
\path [thick,->,>=stealth]

(2) edge [] node[above] {$a$} (1)
(2) edge [bend left=20] node[above] {$a$} (3)
(1) edge [] node[above, pos=0.3] {$a$} (3)

(1) edge [loop below] node[left] {$a$} (1)

(3) edge [bend left=20] node[above] {$b$} (2)
(3) edge [] node[below] {$b$} (4)
(4) edge [] node[left, pos=0.3] {$b$} (2)

(4) edge [loop,out=200,in=160,looseness=8] node[left] {$b$} (4)
;
\end{tikzpicture}
\end{tabular}
\caption{$[a \cdot 3] = M(a) [3] = [\{1, 2\}]$ is a column; $[2 \cdot a]^T = [2]^T M(a) = [\{1, 3\}]^T$ is a row.}\label{fig:ex-columns-rows}
\end{figure}
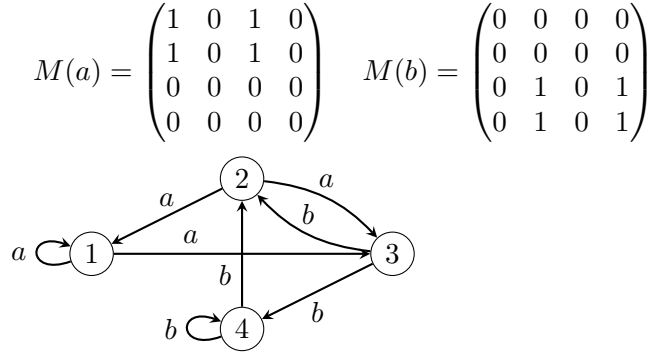

% Fix indentation
\noindent
A column $[C]$ is called \emph{maximal} if there is no column $[C']$ such that $[C'] \ne [C]$ and $[C'] \ge [C]$ (that is, $C \subset C'$). Maximal rows are defined in the same way. Recall that the inequalities are taken component-wise. 

Let $A \in \{0,1\}^{m \times n}$ be a zero-one matrix.
One can view~$\rkr(A)$ as the least number~$r$ such that there are matrices $C \in \mathbb{R}^{m \times r}$ and $R \in \mathbb{R}^{r \times n}$ with $A = C R$.
Define the \emph{unambiguous rank} $\rku(A)$ as the least number~$r$ such that there are matrices $C \in \{0,1\}^{m \times r}$ and $R \in \{0,1\}^{r \times n}$ such that $A = C R$.
Analogously to $\rkr(M)$, define also $\rku(M) := \min\{\rku(M(w)) \mid w \in \Sigma^*\}$.
Clearly, $\rkr(A) \le \rku(A)$, and the inequality can be strict, but in \autoref{cor:ranks-equal} below we show that $\rkr(M) = \rku(M)$.
The reason we are interested in the unambiguous rank is that \autoref{thm:Cesari} below implies that there is always a matrix with a very simple structure such that its unambiguous rank is equal to its real rank and both ranks are minimum.

In the following let us write $r := \rku(M)$ when $M$~is understood.
A word $u \in \Sigma^*$ is \emph{of minimum unambiguous rank} if $\rku(M(u)) = r$.
If $u \in \Sigma^*$ is of minimum unambiguous rank then so is $v u w$ for all $v, w \in \Sigma^*$. 

Words of unambiguous rank one, known as synchronising words, play an especially important role due to their applications in the theory of codes, as explained in \autoref{sec-existing}. It is easy to see that a word $w$ has unambiguous rank one if and only if there exist $C, R \subseteq Q$ such that $w$ maps a state $p$ to a state $q$ if and only if $p \in C$ and $q \in R$.
For total DFAs, we moreover have that $C = Q$ and $R$ has cardinality one.

%Changed to thmC
\begin{thmC}[C{\'{e}}sari \cite{Cesari74}] \label{thm:Cesari}
Let $u \in \Sigma^*$ be of minimum unambiguous rank.
There are pairwise disjoint sets $C_1, \ldots, C_r \subseteq Q$ and pairwise disjoint sets $R_1, \ldots, R_r \subseteq Q$ such that \[M(u) = \sum_{i=1}^r [C_i] [R_i]^T.\]
Moreover, each $[C_i]$ and $[R_i]^T$ is, respectively, a maximal column and a maximal row.
\end{thmC}

% Fix indentation
\noindent
This theorem will play a central role.
A proof can be found in \cite[Proposition 4]{BealCKP08}. For the sake of completeness, we provide an elementary proof of the next two sections.
In the case of a total DFA, $R_1$ is a singleton and \autoref{thm:Cesari} is fairly obvious.

In \autoref{thm:Cesari}, since the $C_i$ are pairwise disjoint and the $R_i$ are pairwise disjoint, each $[C_i] [R_i]^T$ forms, intuitively, a ``combinatorial rectangle'', and no such rectangle shares a row or a column with any other rectangle.
The column vectors~$[C_i]$ are exactly the nonzero columns of~$M(u)$ and linearly independent, and the row vectors~$[R_i]^T$ are exactly the nonzero rows of~$M(u)$ and linearly independent.
Thus, $r$~is the number of distinct nonzero columns and also the number of distinct nonzero rows in~$M(u)$.
It follows that $r = \rku(M(u)) = \rkr(M(u))$. Thus we have:

\begin{cor} \label{cor:ranks-equal}
We have $r = \rku(M) = \rkr(M)$.
\end{cor}

% Fix indentation
\noindent
We can thus define $\rkg(M)$ as $\rku(M) = \rkr(M)$.
For words $w \in \Sigma^*$ that are not of minimum unambiguous rank, we may have $\rkr(M(w)) < \rku(M(w))$, but the rank of such matrices will rarely play a role in the following. In what follows, we call words of minimum unambiguous rank simply words of minimum rank. Since below we never refer to the real rank of words, this will not lead to any confusion.

\subsection{A proof of the first statement in C{\'{e}}sari's theorem}

We use the following lemma.
\begin{lem} \label{lem:Cesari-lemma}
Let $A \in \{0, 1, 2, \ldots\}^{Q \times Q}$ be with $[Q]^T A \ge [Q]^T$ and $A [Q] \ge [Q]$.
Then $A$~is a permutation matrix \emph{or} $\sup_{i \ge 0} [Q]^T A^i [Q] = \infty$.
\end{lem}
\begin{proof}
By a straightforward induction argument we have that $A^{i+1} [Q] \ge A^i [Q]$ holds for all $i \ge 0$.
It suffices to show that $A$~is a permutation matrix \emph{or} for all $i \ge 0$ we have $[Q]^T A^{i+1} [Q] > [Q]^T A^i [Q]$.
We suppose that $A$~is not a permutation matrix and show by induction on~$i$ that $[Q]^T A^{i+1} [Q] > [Q]^T A^i [Q]$ holds for all $i \ge 0$.

Concerning the induction base, $i=0$, since $A$ is not a permutation matrix, we cannot have both $[Q]^T A = [Q]^T$ and $A [Q] = [Q]$ (in fact, we can have neither).
Therefore, $[Q]^T A [Q] > [Q]^T [Q]$.
For the induction step, suppose that $[Q]^T A^{i+1} [Q] > [Q]^T A^i [Q]$ holds for some $i \ge 0$.
Then $A^{i+1} [Q] > A^i [Q]$; i.e., the inequality $A^{i+1} [Q] \ge A^i [Q]$ is strict in some component.
Since $[Q]^T A \ge [Q]^T$, the vector $[Q]^T A$ is strictly positive in all components.
It follows that $[Q]^T A A^{i+1} [Q] > [Q]^T A A^i [Q]$, as required.
\end{proof}

% Fix indentation
\noindent
This allows us to prove the first statement of \autoref{thm:Cesari}. The second statement is proved at the end of the next subsection.

\begin{prop} \label{prop:Cesari-first}
Let $u \in \Sigma^*$ be of minimum rank.
There are pairwise disjoint sets $C_1, \ldots, C_r \subseteq Q$ and pairwise disjoint sets $R_1, \ldots, R_r \subseteq Q$ such that $M(u) = \sum_{i=1}^r [C_i] [R_i]^T$.
\end{prop}
\begin{proof}
Since $\rku(M(u)) = r$, there are $C \in \{0,1\}^{Q \times r}$ and $R \in \{0,1\}^{r \times Q}$ with $M(u) = C R$.
For notational convenience, take a finite set $X$ with $|X| = r$ and write $C \in \{0,1\}^{Q \times X}$ and $R \in \{0,1\}^{X \times Q}$.
Let $w \in \Sigma^*$.
Define $P(w) := R M(w) C \in \{0,1,2,\ldots\}^{X \times X}$.
For all $i \ge 0$ we have $C P(w)^i R = (C R M(w))^i C R = M((u w)^i u) \in \{0,1\}^{Q \times Q}$.
Since $C$~has no zero columns and $R$~has no zero rows, we have $P(w)^i \in \{0,1\}^{X \times X}$ for all~$i$.
Further, since $\rku(P(w)) \ge \rku(C P(w) R) = \rku(M(u w u)) = r$, we have $\rku(P(w)) = r$.
In particular, $P(w)$ does not have a zero row or column; i.e., $P(w) [X] \ge [X]$ and $[X]^T P(w) \ge [X]^T$. 
It follows from \autoref{lem:Cesari-lemma} that $P(w)$~is a permutation matrix.
As $w \in \Sigma^*$ was arbitrary, $P(w)$ is a permutation matrix for all~$w$.

Let $q_1 \in Q$, and let $q_2 \in Q$ and $y \in X$ be such that $[q_2]^T C [y] = 1$.
By strong connectedness, there is $w \in \Sigma^*$ with $[q_1]^T M(w) [q_2] = 1$.
Thus,
\[
P(w) [y] \ = \ R M(w) C [y] \ \ge \ R [q_1] [q_1]^T M(w) [q_2] [q_2]^T C [y] \ = \ R [q_1] \,.
\]
As $P(w)$ is a permutation matrix and $q_1 \in Q$ was arbitrary, it follows that all nonzero columns of~$R$ are of the form $[x]$ for $x \in X$.
Similarly, all nonzero rows of~$C$ are of the form~$[x]^T$ for $x \in X$.

For each $x \in X$, define $C_x, R_x \subseteq Q$ such that $[C_x] = C [x]$ and $[R_x]^T = [x]^T R$.
The~$R_x$ are pairwise disjoint, as if there was $q \in R_x \cap R_y$ with $x \ne y$, then $R [q] \ge [\{x,y\}]$, contradicting what we proved in the previous paragraph.
Similarly, the $C_x$ are pairwise disjoint.
Finally, we have
\[
 M(u) \ = \ C R \ = \ \sum_{x \in X} C [x] [x]^T R \ = \ \sum_{x \in X} [C_x] [R_x]^T\,,
\]
as desired.
\end{proof}

\subsection{The weight of columns and rows}\label{subsec:weights}

The results in this subsection, about the column and row vectors that appear in the matrices~$M(w)$, are mostly due to \cite{Cesari74}; see also \cite[Section 3]{BealCKP08}.
Since a notion of column and row weight will be crucial for us in the later development, we phrase and prove the results around these concepts, but we do not view the lemmas of this subsection as novel.

Define $\overline{A} = \frac{1}{|\Sigma|} \sum_{a\in \Sigma} M(a) \in [0,1]^{Q \times Q}$.
Since $M$ is strongly connected, $\overline{A}$~is irreducible.
Since $M$ is unambiguous, the spectral radius of~$\overline{A}$ is at most~$1$, and since $M$ is complete, it is at least~$1$.
Thus, the spectral radius of~$\overline{A}$ equals~$1$.
Since $\overline{A}$~is irreducible, it follows from basic Perron-Frobenius theory that $\overline{A}$~has an eigenvalue~$1$ and every right eigenvector with eigenvalue~$1$ is a multiple of a strictly positive vector, say $\beta \in \RP^Q$.
Since $\overline{A}$~has only rational entries, we can assume $\beta \in \mathbb{Q}_{>0}^Q$.
Similarly for left eigenvectors.
Therefore, there are $\alpha, \beta \in \mathbb{Q}_{>0}^Q$ with $\alpha^T \overline{A} = \alpha^T$ and $\overline{A} \beta = \beta$.
Without loss of generality, we assume that~$\alpha^T \beta = 1$.

In the total DFA case, since $M(a) [Q] = [Q]$ for all $a \in \Sigma$, we have $\overline{A} [Q] = [Q]$ and so it is natural to take $\beta = [Q]$.
In that case, $\alpha^T [Q] = \alpha^T \beta = 1$ means that $\alpha^T = \alpha^T \overline{A}$ is the (unique) stationary distribution of the Markov chain whose transition probabilities are given by the row-stochastic matrix~$\overline{A}$; intuitively, in this Markov chain a letter $a \in \Sigma$ is picked uniformly at random in every step.

Define the \emph{weight} of a column~$y$ and of a row~$x^T$ by $\alpha^T y \in \mathbb{R}$ and $x^T \beta \in \mathbb{R}$, respectively.
Denote the maximum column weight and the maximum row weight by $\mcw$ and $\mrw$, respectively, i.e.,
\[
 \mcw \ := \ \max\{\alpha^T y \mid y \text{ is a column}\} \quad \text{ and } \quad
 \mrw \ := \ \max\{x^T \beta   \mid x^T \text{ is a row}\}\,.
\]

% Fix indentation
\noindent
A column~$y$ is called \emph{of maximum weight} if $\alpha^T y = \mcw$, and analogously for rows.
In the total DFA case, every row is of the form $[q]^T$ for some $q \in Q$, hence every row is of maximum weight.

For $q \in Q$ define
\[
 \Mer(q) \ := \ \{q' \in Q \mid \exists\,S \supseteq \{q,q'\} \text{ such that } [S] \text{ is a column}\}\,.
\]
Intuitively, $\Mer(q)$ consists of the states that can ``appear'' in a column together with~$q$, or, equivalently, the states that are ``mergeable'' with~$q$ (that is, can be mapped to the same state by a word).
Note that $q \in \Mer(q)$. 

The following lemmas apply symmetrically also to rows.

\begin{lem} \label{lem:max-column}
Let $v \in \Sigma^*$ and $q \in Q$ be such that $[v \cdot q]$ is a column of maximum weight. Then we have:
\begin{enumerate}
\item[(a)] $[u v \cdot q]$ is a column of maximum weight for all $u \in \Sigma^*$;
\item[(b)] $v \cdot q' = \emptyset$ holds for all $q' \in \Mer(q) \setminus \{q\}$.
\end{enumerate}
\end{lem}

\begin{proof}
Towards~(a), let $a \in \Sigma$.
It suffices to prove that $M(a) [v \cdot q]$ is of maximum weight.
We have $\alpha^T [v \cdot q] = \alpha^T \overline{M} [v \cdot q] = \frac{1}{|\Sigma|} \sum_{b \in \Sigma} \alpha^T M(b)[v \cdot q]$.
Thus, if $\alpha^T M(a) [v \cdot q] < \alpha^T [v \cdot q]$, then there would also exist $b \in \Sigma$ with $\alpha^T M(b) [v \cdot q] > \alpha^T [v \cdot q]$, contradicting that $[v \cdot q]$ is of maximum weight.
So $M(a)[v \cdot q]$ is of maximum weight.

Towards~(b), let $q' \in \Mer(q) \setminus \{q\}$.
Since $M$ is strongly connected, there is $w \in \Sigma^*$ with $[w \cdot q] \ge [q] + [q']$.
Thus, $[v w \cdot q] \ge [v \cdot q] + [v \cdot q']$.
It follows that, since $[v \cdot q]$ is of maximum weight, so is $[v w \cdot q]$.
Hence, $[v \cdot q'] = 0$.
\end{proof}

\begin{lem} \label{lem:non-max-column}
Let $S \subseteq Q$ be such that $[S]$ is a column which is not of maximum weight.
\begin{enumerate}
\item[(a)] There is $S' \supsetneq S$ such that $[S']$ is a column of maximum weight.
\item[(b)] There is $u \in \Sigma^*$ with $M(u) [S] = 0$.
\end{enumerate}
\end{lem}

\begin{proof}
Let $v \in \Sigma^*$ and $q \in Q$ be such that $[v \cdot q]$ is not of maximum weight.

Towards~(a), let $w \in \Sigma^*$ and $t \in Q$ be such that $[w \cdot t]$ is of maximum weight and $w \cdot t \ni q$.
By \autoref{lem:max-column}~(a), $[v w \cdot t]$ is of maximum weight.
Since $v w \cdot t \supseteq v \cdot q$ and $[v \cdot q]$ is not of maximum weight, we have $v w \cdot t \supsetneq v \cdot q$.

Towards~(b), let $p \in S' \setminus S$.
We have $S \subseteq S' \subseteq \Mer(p)$ and thus $S \subseteq \Mer(p) \setminus \{p\}$.
Let~$u \in \Sigma^*$ be such that $[u \cdot p]$ is of maximum weight.
Then it follows from \autoref{lem:max-column}~(b) that $M(u) [S] = 0$.
\end{proof}

\begin{lem} \label{lem:max-notions}
    A column (respectively, row) is maximal if and only if it is of maximum weight.
\end{lem}
\begin{proof}
If a column is of maximum weight, it is clearly maximal.
Conversely, if a column is not of maximum weight, then \autoref{lem:non-max-column}~(a) says it is not maximal.
\end{proof}

% Fix indentation
\noindent
An important property that we will need later is that the set of maximal columns is closed under left multiplication by matrices from the monoid, as stated in the following lemma. Note that this is no longer true without the completeness assumption, and is the key reason why the case of complete matrix monoids is easier to deal with.

\begin{lem} \label{lem:max-column-stable}
Let $v \in \Sigma^*$ and $q \in Q$ be such that $[v \cdot q]$ is a maximal column. Then $[u v \cdot q]$ is a maximal column for all $u \in \Sigma^*$.
\end{lem}
\begin{proof}
Immediate from \autoref{lem:max-column}~(a) and \autoref{lem:max-notions}.
\end{proof}

% Fix indentation
\noindent
The following lemma will be useful later to construct minimum rank matrices from maximal columns and rows.

\begin{lem} \label{lem:all-max-column}
Let $w \in \Sigma^*$ be such that all nonzero columns and rows in $M(w)$ are maximal. Then %$M(ww)$
$w w$ is of minimum rank.
\end{lem}
\begin{proof}
Define $S := \{q \in Q \mid w \cdot q \ne \emptyset \ne q \cdot w\}$.
Since
\[
M(w w) \ = \ M(w) I M(w) \ = \ M(w) \big( \sum_{q \in Q} [q] [q]^T \big) M(w) \ = \ \sum_{q \in Q} [w \cdot q] [q \cdot w]^T \ = \ \sum_{q \in S} [w \cdot q] [q \cdot w]^T\,,
\]
we have $\rku M(w w) \le |S|$. %$ \le \rku M(u w w)$.

% Fix indentation
\noindent
Let $u \in \Sigma^*$ be of minimum rank.
Suppose there were a state $p \in Q$, two states $q, q' \in S$ and $u w$-labelled paths from $p$ to~$q$ and also from $p$ to~$q'$.
Then $p \cdot u w w \supseteq q \cdot w \cup q' \cdot w$, contradicting the maximality of the row $[q \cdot w]^T$.
Therefore, the sets $u w \cdot q$, where $q \in S$, are pairwise disjoint.
Since the columns $[w \cdot q]$ for $q \in S$ are maximal, by \autoref{lem:max-column-stable} the columns $[u w \cdot q]$ are also maximal and in particular nonzero.
Moreover, these columns appear in $M(u w w)$, as $q \cdot w \ne \emptyset$ for all $q \in S$.
Therefore, $|S| \le \rku M(u w w)$.
Since $u$ is of minimum rank, we have $\rku M(u w w) \le r$.
By combining all inequalities we obtain that $\rku M(w w) \le r$.
Hence, $w w$ is of minimum rank.
\end{proof}

% Fix indentation
\noindent
We can now complete the proof of \autoref{thm:Cesari}.
\begin{proof}[Proof of \autoref{thm:Cesari}]
The first statement is \autoref{prop:Cesari-first}.
We prove the second statement only for columns, as the proof for rows is analogous.
Towards a contradiction, suppose that~$[C_i]$ is a non-maximal column; without loss of generality, say $i=r$.
By \autoref{lem:max-notions}, $[C_i]$ is not of maximum weight.
By \autoref{lem:non-max-column}~(b), there is $v \in \Sigma^*$ such that $M(v) [C_r] = 0$.
Then 
\[
 M(v u) \ = \ \sum_{i=1}^r M(v) [C_i] [R_i]^T \ = \ \sum_{i=1}^{r-1} M(v) [C_i] [R_i]^T \ = \ C R^T\,,
\]
where 
$C, R \in \{0,1\}^{Q \times (r-1)}$ are the matrices whose $i$th columns are $M(v) [C_i]$ and $[R_i]$, respectively.
Hence, $\rku(M(v u)) \le r-1$, contradicting the definition of~$r$.
\end{proof}

\subsection{Weight preservation property and minimum rank}\label{subsec:weight-preserv}

Every word $w$ of minimum rank in a total DFA induces a partition of the state set into subsets of states mapped by $w$ to the same state (that is, into columns). It was observed by Friedman in \cite{Friedman1990} that all sets of such a partition have the same weight. This observation has many applications to variations of the road colouring problem \cite{Friedman1990,Kari2001,Gusev2016,Kari2019}. Moreover, it was proved in \cite[Theorem 6]{Gusev2016}, again in connection with road colouring, that for every $w$ the weights of all columns in $M(w)$ sum up to $1$ (assuming $\beta = [Q]$ as suggested previously). 
This can be seen as a weight preservation property: the total weight of columns in the matrix of a word is preserved under multiplication by any matrix from the monoid.
As a result we get that $1 = r \cdot \mcw$, and hence~$r = \frac{1}{\mcw}$. The proof of the weight preservation property for total DFAs is quite simple and relies on the fact that for a state $q$ and a word~$w$ the set~$q \cdot w$ is always a singleton. For complete UFAs this is no longer true; in particular, $q \cdot w$ can be the empty set, thus permanently ``losing'' some weight collected in $q$. Hence, a more refined property is required. The following result provides such a property.
It also turns out that its proof requires more sophisticated techniques than in the total DFA case.

\begin{thm} \label{thm:rank-weight}
For all $w \in \Sigma^*$ we have $1 = \alpha^T M(w) \beta = r \cdot \mcw \cdot \mrw$.
\end{thm}

% Fix indentation
\noindent
Similarly to the total DFA case, this result allows us to reduce computing~$r$ to computing $\mcw$ and~$\mrw$, which we will use later in our algorithms. Recall that we have defined $\alpha^T$ and~$\beta$ so that $\alpha^T \beta = 1$.

Towards a proof of \autoref{thm:rank-weight} we first prove the following lemma.

\begin{lem} \label{lem:minrank-weight}
Let $u \in \Sigma^*$ be of minimum rank.
Then $\alpha^T M(u) \beta = r \cdot \mcw \cdot \mrw$.
\end{lem}
\begin{proof}
Let $M(u) = \sum_{i=1}^r [C_i] [R_i]^T$ be as in \autoref{thm:Cesari}.
Each $[C_i]$ and each $[R_i]$ is of maximum weight.
Thus,
\[
 \alpha^T M(u) \beta \ = \ \sum_{i=1}^r \alpha^T [C_i] [R_i]^T \beta \ = \ \sum_{i=1}^r \mcw \cdot \mrw \ = \ r \cdot \mcw \cdot \mrw\,. \qedhere
\]
\end{proof}

% Fix indentation
\noindent
We also need the following proposition.

\begin{prop}\label{prop:compact}
Let $x \in \mathbb{R}^{Q}$ and $c \in \mathbb{R}$ be such that $x^T M(u) \beta = c$ holds for all $u \in \Sigma^*$ of minimum rank.
Then $x^T M(w) \beta = c$ holds for all $w \in \Sigma^*$.
\end{prop}

\begin{proof}
Let $w \in \Sigma^*$.
Let $u \in \Sigma^*$ be of minimum rank.
Recall that every word that contains~$u$ as a factor is of minimum rank.
For every $k \ge 0$, partition $\Sigma^k$ into sets $W_0(k)$ and $W_1(k)$ so that $W_0(k) = \Sigma^k \cap (\Sigma^* u \Sigma^*)$ and $W_1(k) = \Sigma^k \setminus (\Sigma^* u \Sigma^*)$; i.e., $W_0(k), W_1(k)$ are the sets of length-$k$ words that do or do not contain~$u$ as a factor, respectively.
For all $v \in W_0(k)$ both~$v$ and $w v$ are of minimum rank.
Thus, we have $x^T M(w v) \beta = c$ for all $v \in W_0(k)$.
It follows that
\begin{equation} \label{eq:avg-W0}
\sum_{v \in W_0(k)} \frac{x^T M(w v) \beta}{|W_0(k)|} \ = \ c \quad \text{for all $k \ge 0$.} 
\end{equation}
Let $d > 0$ be such that $|x^T A \beta| \le d$ for all $A \in \{0,1\}^{Q \times Q}$.
Then we have
\begin{equation} \label{eq:avg-W1}
\sum_{v \in W_1(k)} \frac{|x^T M(w v) \beta|}{|W_1(k)|} \ \le \ d \quad \text{for all $k \ge 0$.} 
\end{equation}

% Fix indentation
\noindent
Let $m \ge 0$.
Define $p_1(m) := \frac{|W_1(m |u|)|}{|\Sigma|^{m |u|}}$.
We can view~$p_1(m)$ as the probability of picking a word in~$W_1(m |u|)$ when a word of length~$m |u|$ is picked uniformly at random.
We have
$p_1(m) \le \left( 1 - \frac{1}{|\Sigma|^{|u|}} \right)^m$, as in order to avoid $u$ as a factor, it has to be avoided in each of the $m$ consecutive blocks of length~$|u|$.
Thus, $\lim_{m \to \infty} p_1(m) = 0$.
We have
\begin{align*}
x^T M(w) \beta
 \ &= \ x^T M(w) \overline{A} \beta \ = \ x^T M(w) \overline{A}^{m |u|} \beta \ = \ \frac{1}{|\Sigma|^{m |u|}} \sum_{v \in \Sigma^{m |u|}} x^T M(w v) \beta \\
 \ &= \ \frac{|W_0(m |u|)|}{|\Sigma|^{m |u|}} \sum_{v \in W_0(m |u|)} \frac{x^T M(w v) \beta}{|W_0(m |u|)|} \ + \ \frac{|W_1(m |u|)|}{|\Sigma|^{m |u|}} \sum_{v \in W_1(m |u|)} \frac{x^T M(w v) \beta}{|W_1(m |u|)|} \\
 \ &= \ (1-p_1(m)) \cdot c + p_1(m) \cdot \sum_{v \in W_1(m |u|)} \frac{x^T M(w v) \beta}{|W_1(m |u|)|} \qquad (\text{by \autoref{eq:avg-W0}).}
\end{align*}
With~\autoref{eq:avg-W1} it follows that
\[
 |x^T M(w) \beta - c| \ \le \ p_1(m) ( |c| + d)\,.
\]
Since this holds for all $m \ge 0$ and $\lim_{m \to \infty} p_1(m) = 0$, we conclude that $x^T M(w) \beta = c$.
\end{proof}

% Fix indentation
\noindent
Now we prove \autoref{thm:rank-weight}.

\begin{proof}[Proof of \autoref{thm:rank-weight}]
It follows from \autoref{lem:minrank-weight} and \autoref{prop:compact} that
\[
 \alpha^T M(w) \beta \ = \ r \cdot \mcw \cdot \mrw \qquad \text{for all $w \in \Sigma^*$.}
\]
With $w = \varepsilon$ we obtain $1 = \alpha^T \beta = \alpha^T M(\varepsilon) \beta = r \cdot \mcw \cdot \mrw$, as required.
\end{proof}

\subsection{Maximal pseudo-columns}\label{subsec:pseudocolumns}

In this subsection, we define maximal pseudo-columns, which are vectors that can be seen as a relaxation of the notion of maximal columns. We show that the weight of a maximal pseudo-column is equal to the weight of a maximal column, and a maximal pseudo-column is a solution of a system of linear equations, and thus can be computed efficiently. By invoking \autoref{thm:rank-weight}, this will allow us to efficiently compute $r$.

Denote by $\MCol \subseteq \{0,1\}^Q$ the set of maximal columns.
By \autoref{thm:Cesari} (bearing in mind also \autoref{cor:ranks-equal}), the vector space spanned by all maximum columns, $\spn{\MCol}$, is at least $r$-dimensional:
\begin{prop} \label{prop:r<=MC}
We have $r \le \dim \spn{\MCol}$.
\end{prop}

% Fix indentation
\noindent
One might hypothesise that $r = \dim \spn{\MCol}$ or even that all minimum-rank matrices have the same $r$ nonzero (hence, maximum) columns.
The following example shows that neither is the case in general, not even for total DFAs.

\begin{exa} \label{ex:snap1}
Consider the total DFA with $\Sigma = \{a,b\}$ and
\begin{figure}[H]
    \centering
    \begin{subfigure}[c]{0.65\textwidth} \centering
\[
M(a) = \begin{pmatrix} 1 & 0 & 0 & 0 \\ 0 & 0 & 1 & 0 \\ 0 & 0 & 1 & 0 \\ 1 & 0 & 0 & 0 \end{pmatrix} \quad
M(b) = \begin{pmatrix} 0 & 1 & 0 & 0 \\ 0 & 1 & 0 & 0 \\ 0 & 0 & 0 & 1 \\ 0 & 0 & 0 & 1 \end{pmatrix}\]
\end{subfigure}
\begin{subfigure}[c]{0.30\textwidth} \centering
\begin{tikzpicture}[node distance = 2cm]
%,
\tikzset{every state/.style={inner sep=1pt,minimum size=1.5em}}
\node[state] (1) at (0,0) {$1$};
\node[state] (2) at (1,0) {$2$};
\node[state] (3) at (1,-1) {$3$};
\node[state] (4) at (0,-1) {$4$};
\path [-stealth, thick] (1) edge [loop,out=155,in=115,looseness=8] node[left] {$a$} (1)
(2) edge [loop,out=65,in=25,looseness=8] node[right] {$b$} (2)
(3) edge [loop,out=-25,in=-65,looseness=8] node[right] {$a$} (3)
 (4) edge [loop,out=-115,in=-155,looseness=8] node[left] {$b$} (4)
 (1) edge node[above] {$b$} (2)
(2) edge node[right] {$a$} (3)
(3) edge node[below] {$b$} (4)
(4) edge node[left] (b) {$a$} (1);
\end{tikzpicture}
\end{subfigure}
\end{figure}

% Fix indentation
\noindent
By symmetry, we have $\alpha^T = \begin{pmatrix} \frac14 & \frac14 & \frac14 & \frac14 \end{pmatrix}$.
Since no word maps states $1$ and~$3$ to the same state, we have $r=2$; i.e., $a$ and $b$ are both minimum rank.
Further, $\MCol$ consists exactly of the four nonzero columns in $M(a)$ and~$M(b)$. 
Their span $\spn{\MCol}$ is $3$-dimensional, as $\begin{pmatrix}1 & -1 & 1 & -1\end{pmatrix}$ is orthogonal to each maximum column.
Thus, $r = 2 < 3 = \dim \spn{\MCol} < 4 = |\MCol|$.
\end{exa}

% Fix indentation
\noindent
Define the vector space $U := \spn{\alpha^T M(w) - \alpha^T \mid w \in \Sigma^*}$. Intuitively, it is the set of all differences of weight distributions over the states before and after a word is applied.
Notice that for all $w_1, w_2 \in \Sigma^*$ we have $\alpha^T M(w_1) - \alpha^T M(w_2) \in U$.
Later (see the proof of \autoref{lem:U-in-NC} below) we show that $U$ is closed under post-multiplication with $M(a)$ for all $a \in \Sigma$.
Such ``forward spaces'' play an important role in weighted automata; see, e.g.,~\cite{KieferWeighted20}.
Denote the orthogonal complement of~$U$ by~$U^\bot$; i.e., $U^\bot = \{y \in \R^Q \mid \forall\,w \in \Sigma^*: \ \alpha^T M(w) y = \alpha^T y\}$. Intuitively, it is the set of vectors whose weight does not change under pre-multiplication with $M(w)$ for any $w$ (where by the weight of a vector $y$ we understand $\alpha^T y$).
Clearly, $\dim U + \dim U^\bot = |Q|$.
The following proposition follows immediately from \autoref{lem:max-column-stable}.

\begin{prop} \label{prop:MC<=V}
We have $\MCol \subseteq U^\bot$.
\end{prop}

% Fix indentation
\noindent
It follows that $\spn{\MCol}$ is a subspace of~$U^\bot$.
With \autoref{prop:r<=MC}, we have $r \le \dim \spn{\MCol} \le \dim U^\bot$.
One might hypothesise that $\spn{\MCol} = U^\bot$.
The following example  shows that this is not the case in general, not even for total DFAs.

\begin{exa}\label{ex-appendix}
Consider the DFA with $\Sigma = \{a,b,c\}$ and
\[
M(a) = \begin{pmatrix} 1 & 0 & 0 & 0 \\ 1 & 0 & 0 & 0 \\ 0 & 0 & 1 & 0 \\ 0 & 0 & 1 & 0 \end{pmatrix}, \, %\quad
M(b) = \begin{pmatrix} 0 & 1 & 0 & 0 \\ 0 & 1 & 0 & 0 \\ 0 & 0 & 0 & 1 \\ 0 & 0 & 0 & 1 \end{pmatrix}, \, %\quad
M(c) = \begin{pmatrix} 0 & 0 & 1 & 0 \\ 0 & 0 & 0 & 1 \\ 1 & 0 & 0 & 0 \\ 0 & 1 & 0 & 0 \end{pmatrix}\,.
\]
\begin{center}
\begin{tikzpicture}[node distance = 2cm]
%,
\tikzset{every state/.style={inner sep=1pt,minimum size=1.5em}}
\node[state] (1) at (0,0) {$1$};
\node[state] (2) at (1.5,0) {$2$};
\node[state] (3) at (0,-1.5) {$3$};
\node[state] (4) at (1.5,-1.5) {$4$};

\path [-stealth, thick] (1) edge [loop,out=155,in=115,looseness=8] node[left] {$a$} (1)
(2) edge [loop,out=65,in=25,looseness=8] node[right] {$b$} (2)
 (3) edge [loop,out=-115,in=-155,looseness=8] node[left] {$a$} (3)
 (4) edge [loop,out=-25,in=-65,looseness=8] node[right] {$b$} (4)
 (2) edge[bend left=15] node[below] {$a$} (1)
 (1) edge[bend left=15] node[above] {$b$} (2)
 (4) edge[bend left=15] node[below] {$a$} (3)
(3) edge[bend left=15] node[above] {$b$} (4)
 (1) edge[bend left=15] node[right] {$c$} (3)
(3) edge[bend left=15] node[left] {$c$} (1)
(2) edge[bend left=15] node[right] {$c$} (4)
(4) edge[bend left=15] node[left] (b) {$c$} (2);
\end{tikzpicture}
\end{center}

% Fix indentation
\noindent
We have $\Mer(1) = \Mer(2) = \{1,2\}$ and $\Mer(3) = \Mer(4) = \{3,4\}$.
Thus, $\MCol = \{\begin{pmatrix}1 & 1 & 0 & 0\end{pmatrix}^T,\begin{pmatrix}0 & 0 & 1 & 1\end{pmatrix}^T\}$.
Hence, $\dim \spn{\MCol} = 2$.

On the other hand, by symmetry we have $\alpha^T = \begin{pmatrix} \frac14 & \frac14 & \frac14 & \frac14 \end{pmatrix}$.
For any $w \in \Sigma^*$,
\[
\alpha^T M(w) \ = \ \begin{cases}
 \begin{pmatrix} \frac14 & \frac14 & \frac14 & \frac14 \end{pmatrix} & \text{if } w \in \{c\}^* \\
 \begin{pmatrix} \frac12 & 0 & \frac12 & 0 \end{pmatrix} & \text{if the last non-$c$ letter in $w$ is $a$} \\
 \begin{pmatrix} 0 & \frac12 & 0 & \frac12 \end{pmatrix} & \text{if the last non-$c$ letter in $w$ is $b$\,.}
\end{cases}
\]
It follows that $U = \spn{\begin{pmatrix} 1 & -1 & 1 & -1 \end{pmatrix}}$.
Thus, $\dim U^\bot = 4 - 1 = 3 > 2 = \dim \spn{\MCol}$, and $\spn{\MCol}$ is a strict subspace of~$U^\bot$.
For example, the vector $\begin{pmatrix} 1 & 0 & 0 & 1 \end{pmatrix}^T$ is in $U^\bot$ but not in $\spn{\MCol}$.
\end{exa}

% Fix indentation
\noindent
Although the dimension of~$U^\bot$ does not generally equal~$r$, the vector space~$U^\bot$ turns out useful for computing~$r$.
Recall that, by \autoref{thm:rank-weight}, we can obtain~$r$ by computing~$\mcw$ (and, symmetrically,~$\mrw$). Recall that we define for $q \in Q$
\[
 \Mer(q) \ := \ \{q' \in Q \mid \exists\,S \supseteq \{q,q'\} \text{ such that } [S] \text{ is a column}\}\,.
\]
We will need the following lemma which is easy to prove.

\begin{lem}\label{lem:max-column-body}
Let $v \in \Sigma^*$ and $q \in Q$ be such that $[v \cdot q]$ is a maximal column. Then $v \cdot q' = \emptyset$ holds for all $q' \in \Mer(q) \setminus \{q\}$.
\end{lem}

% Fix indentation
\noindent
We call a vector $y \in U^\bot$ a \emph{maximal pseudo-column} if there is $q \in Q$ with $y(q) = 1$ and $y(q') = 0$ for all $q' \not\in \Mer(q)$.
This notion, which is closely related to the ``pseudo-cuts'' from \cite{KieferW19}, can be seen as a relaxation of the notion of a maximal column: clearly, every maximal column is a maximal pseudo-column, but the converse is not true, since a maximal pseudo-column is not necessarily a vector over $\{0, 1\}$, let alone a \emph{column} in the strict sense, i.e., of the form~$[w \cdot p]$. The following lemma however shows that the weight of a maximal pseudo-column is equal to the weight of a maximal column. We will later show that computing the former can be done in \NCT.

\begin{lem} \label{lem:Cas}
Let $y$ be a maximal pseudo-column.
Then $\alpha^T y = \mcw$.
\end{lem}
\begin{proof}
Let $q \in Q$ be such that $y(q) = 1$ and $y(q') = 0$ for all $q' \not\in \Mer(q)$.
Let $w \in \Sigma^*$ be such that $[w \cdot q]$ is a maximal column.
We have
\begin{align*}
\alpha^T y \
&= \ \alpha^T M(w) y && (y \in U^\bot) \\
&= \ \sum_{q' \in Q} y(q') \alpha^T [w \cdot q'] \\
&= \ \sum_{q' \in \Mer(q)} y(q') \alpha^T [w \cdot q'] && (y(q') = 0 \text{ for } q' \not\in \Mer(q)) \\
&= \ \alpha^T [w \cdot q] + \sum_{q' \in \Mer(q) \setminus \{q\}} y(q') \alpha^T [w \cdot q'] && (y(q) = 1) \\
&= \ \alpha^T [w \cdot q] && (\text{\autoref{lem:max-column-body}}) \\
&= \ \mcw && (\text{by the choice of $w, q$}). \qedhere
\end{align*}
\end{proof}

\begin{exa} \label{ex:snap2}
We continue \autoref{ex:snap1}.
We have $U = \spn{\begin{pmatrix} 1 & -1 & 1 & -1 \end{pmatrix}}$ and $\Mer(2) = \{1,2,3\}$.
Let $y = \begin{pmatrix}4/3 & 1 & -1/3 & 0\end{pmatrix}^T$.
Then $y \in U^\bot$.
Since $y(2) = 1$ and $y(4) = 0$, vector~$y$ is a maximal pseudo-column.
Thus, by \autoref{lem:Cas}, $\mcw = \alpha^T y = \begin{pmatrix} \frac14 & \frac14 & \frac14 & \frac14 \end{pmatrix} y = \frac12$.
\end{exa}

\begin{thm}\label{thm:pseud-linsyst}
Let $\Gamma$ be a basis of $U$, and let $q \in Q$.
Then the following linear system for $y \in \mathbb{R}^Q$ has a solution, and all its solutions are maximal pseudo-columns:
\begin{align*}
 \gamma^T y \ &= \ 0 \quad \text{for all } \gamma^T \in \Gamma \\
 y(q) \ &= \ 1 \\
 y(q') \ &= \ 0 \quad \text{for all } q' \not \in \Mer(q)\,.
\end{align*}
\end{thm}
\begin{proof}
By \autoref{prop:MC<=V}, any maximal column solves the linear system.
Let $y \in \mathbb{R}^Q$ be a solution of the linear system.
The equations on the first line guarantee that $y \in U^\bot$.
Then, the equations on the second and third line guarantee that $y$~is a maximal pseudo-column.
\end{proof}

\subsection{Dealing with the non-strongly connected case}\label{subsec:non-sc}

The following lemma shows that in order to compute the minimum rank we can focus on the strongly connected case.

\begin{prop} \label{prop:unamb-sc}
Let $M : \Sigma \to \{0,1\}^{Q \times Q}$ be an unambiguous matrix monoid morphism.
Suppose that $Q_1 \cup Q_2 = Q$ is a partition of~$Q$ such that for all $w \in \Sigma^*$ it holds that $[Q_2]^T M(w) [Q_1] = 0$; i.e., for all $w \in \Sigma^*$ matrix~$M(w)$ has the block form
$
M(w) = \begin{pmatrix}
    M_1(w) & M_{12}(w) \\ 0 & M_2(w)
\end{pmatrix}\,,
$
where $M_1(w) \in \{0,1\}^{Q_1 \times Q_1}$ and $M_{12}(w) \in \{0,1\}^{Q_1 \times Q_2}$ and $M_2(w) \in \{0,1\}^{Q_2 \times Q_2}$.
%For each $\rkg \in \{\rk, \rku, \rkr\}$ we have $\rkg(M) = \rkg(M_1) + \rkg(M_2)$.
We have $\rkr(M) = \rkr(M_1) + \rkr(M_2)$ and $\rku(M) = \rku(M_1) + \rku(M_2)$. %whenever~$\rkg \in \{\rk,\rku,\rkr\}$.
\end{prop}
\begin{proof}
It is a well-known property of upper-triangular block matrices that $\rkg(M(w)) \ge \rkg(M_1(w)) + \rkg(M_2(w))$; see, e.g., \cite[Chapter 0.9.4]{HornJohnson13}.
Concerning the analogous property of~$\rku$, let $M(w) = C R$ for some $0/1$ matrices $C, R$.
Then $R$~has at least $\rku(M_2(w))$ rows whose $Q_1$-components are all~$0$, and at least $\rku(M_1(w))$ rows whose $Q_1$-components are not all~$0$.
It follows that $R$ has at least $\rku(M_1(w)) + \rku(M_2(w))$ rows.
Since the factorization $M(w) = C R$ was arbitrary, we conclude that \[\rku(M(w)) \ge \rku(M_1(w)) + \rku(M_2(w)).\]

\newcommand{\myrank}{\mathsf{rk}}
Now let $\myrank \in \{\rkg, \rku\}$ and $w_1, w_2 \in \Sigma^*$ be such that $\myrank(M_1(w_1)) = r_1$ and $\myrank(M_2(w_2)) = r_2$.
With what we have shown in the first paragraph, it suffices to show that $\myrank(M(w_1 w_2)) \le r_1 + r_2$.
For some matrices $A_1, A_2, B_1, B_2$ we have
\begin{align*}
M(w_1 w_2) \ = \ M(w_1) M(w_2) \ &= \ 
\begin{pmatrix}
    M_1(w_1) & A_1 \\ 0 & A_2
\end{pmatrix}
\begin{pmatrix}
    B_1 & B_2 \\ 0 & M_2(w_2)
\end{pmatrix} \\
&= \ 
\begin{pmatrix}
    M_1(w_1) \\ 0 
\end{pmatrix}
\begin{pmatrix}
    B_1 & B_2 
\end{pmatrix}
+
\begin{pmatrix}
    A_1 \\ A_2
\end{pmatrix}
\begin{pmatrix}
   0 & M_2(w_2)
\end{pmatrix}\,.
\end{align*}
The first summand has rank at most $\myrank \begin{pmatrix}
    M_1(w_1) \\ 0 
\end{pmatrix} = r_1$, and similarly the second summand has rank at most~$r_2$.
So we can write both ($i=1,2$) summands as $C_i R_i$, where $C_i$ has at most $r_i$ columns.
Hence, 
\[
M(w_1 w_2) \ = \ C_1 R_1 + C_2 R_2 \ = \ 
\begin{pmatrix}
    C_1 & C_2 
\end{pmatrix}
\begin{pmatrix}
    R_1 \\ R_2
\end{pmatrix}\,.
\]
It follows that $\myrank(M(w_1 w_2)) \le r_1+r_2$.
\end{proof}

% Fix indentation
\noindent
By a straightforward induction it follows from \autoref{prop:unamb-sc} that the minimum rank of an unambiguous matrix monoid is the sum of the minimum ranks of its strongly connected components (where ``incomplete'' components count as having rank~$0$).

%% file: sec-nc2-new.tex
In this section, we prove our first main result, which is as follows.

\begin{thm} \label{thm:rank-in-NC}
The problem of computing the (real) rank of an unambiguous matrix monoid is in~$\NC^2$.
\end{thm}

% Fix indentation
\noindent
In order to use \autoref{thm:pseud-linsyst}, we need the following lemma. We use the notation defined in the previous section.
Recall that we defined $U := \spn{\alpha^T M(w) - \alpha^T \mid w \in \Sigma^*}$.

\begin{lem} \label{lem:U-in-NC}
If $M$ is strongly connected,
one can compute a basis of~$U$ in~\NCT.
\end{lem}

% Fix indentation
\noindent
For each $a \in \Sigma$ define $M'(a) := M(a) - I \in \{-1,0,1\}^{Q \times Q}$ and extend~$M'$ to $M' : \Sigma^* \to \mathbb{Z}^{Q \times Q}$ by defining $M'(a_1 \cdots a_k) = M'(a_1) \cdots M'(a_k)$.
Define \[U' := \spn{\alpha^T M'(w) \mid w \in \Sigma^+}.\]
Note that here $w$ ranges over~$\Sigma^+$, i.e., nonempty words, only. By definition, $U'$ is closed under right multiplication by $M'(a)$ for all $a \in \Sigma$.
We first show the following lemma.

\begin{lem} \label{lem:U=U'}
We have $U = U'$.
\end{lem}
\begin{proof}
For the inclusion $U \subseteq U'$, we prove by induction on $i \ge 0$ that for all length-$i$ words $w \in \Sigma^i$ we have $\alpha^T (M(w) - I) \in U'$.
Concerning the induction base, $i=0$, we have $\alpha^T (M(\varepsilon) - I) = 0 \in U'$.
Concerning the induction step, let $i \ge 0$, and let $w \in \Sigma^i$ and $a \in \Sigma$.
We have
\begin{align*}
\alpha^T (M(w a) - I)
\ &= \ \alpha^T (M(w) - I) M(a) + \alpha^T (M(a) - I) \\
\ &= \ \alpha^T (M(w) - I) (M(a) - I) + \alpha^T (M(w) - I) + \alpha^T (M(a) - I) \\
\ &= \ \alpha^T (M(w) - I) M'(a) + \alpha^T (M(w) - I) + \alpha^T M'(a)\,.
\end{align*}
It holds that $\alpha^T M'(a) \in U'$, and, by the induction hypothesis, $\alpha^T (M(w) - I) \in U'$.
It follows that $\alpha^T (M(w a) - I) \in U'$.

For the converse, $U' \subseteq U$, we proceed similarly by induction.
Concerning the induction base, $i=1$, for all $a \in \Sigma$ we have $\alpha^T M'(a) = \alpha^T (M(a) - I) \in U$.
Concerning the induction step, let $i \ge 1$, and let $w \in \Sigma^i$ and $a \in \Sigma$.
By the induction hypothesis there are $n \le |Q|$ and $w_1, \ldots, w_n \in \Sigma^*$ and $\lambda_1, \ldots, \lambda_n \in \mathbb{R}$ such that $\alpha^T M'(w) = \sum_{i=1}^n \lambda_i \alpha^T (M(w_i) - I)$.
Thus, we have
\begin{align*}
\alpha^T M'(w a)
\ &= \ \alpha^T M'(w) (M(a) - I) 
\ = \ \sum_{i=1}^n \lambda_i \alpha^T (M(w_i) - I) (M(a) - I) \\
\ &= \ \sum_{i=1}^n \lambda_i \alpha^T \big((M(w_i a) - I) - (M(a) - I) - (M(w_i) - I) \big) \in \ U\,,
\end{align*}
as required.
\end{proof}

\begin{proof}[Proof of \autoref{lem:U-in-NC}]
For each $a \in \Sigma$ define $U'_a := \spn{\alpha^T M'(a) M'(w) \mid w \in \Sigma^*}$.
Using the technique from \cite[Section 4.2]{KieferMW14} (see \cite[Proposition 5.2]{KieferWeighted20} for a clearer explanation), for each $a \in \Sigma$ one can compute\footnote{In \cite{KieferMW14,KieferWeighted20} only membership in~$\NC$ is claimed, but the bottleneck computations are matrix powering and rank computation, which can in fact be done in $\mathrm{DET} \subseteq \NC^2$; see~\cite{Cook85}.
The computations in~\cite[Proposition 5.2]{KieferWeighted20} are on polynomially larger matrices, but this does not impact the membership in~$\NC^2$, as $\log^2(\mathit{poly}(n)) = O(\log^2(n))$.} a basis of~$U'_a$ in~$\NC^2$.
The union of these bases, say $\Gamma = \{\gamma_1^T, \ldots, \gamma_n^T\}$ for some $n \le |\Sigma| |Q|$, spans~$U'$, which equals~$U$ by \autoref{lem:U=U'}.
To shrink~$\Gamma$ to a basis of~$U$, for each $i \in \{1, \ldots, n\}$ include~$\gamma_i^T$ in the basis if and only if $\dim \spn{\gamma_1^T, \ldots, \gamma_{i-1}^T} < \dim \spn{\gamma_1^T, \ldots, \gamma_i^T}$.
The latter (rank) computation can be done in~$\NC^2$~\cite{IbarraMoranRosier80}. %; see also \cite[Section~5]{Cook85}.
\end{proof}

% Fix indentation
\noindent
Now we can prove \autoref{thm:rank-in-NC}.
\begin{proof}[Proof of \autoref{thm:rank-in-NC}]
Let $M : \Sigma \to \{0,1\}^{Q \times Q}$ be an unambiguous monoid morphism.
Its strongly connected components can be computed in $\NL \subseteq \NC^2$.
It follows from the proof of \cite[Proposition 3]{Kiefer2021} that one can check each component for completeness in~$\NC^2$, since a zero-one monoid contains the zero matrix if and only if the joint spectral radius of the set of its generators is strictly less than one \cite{Kiefer2021}.
Therefore, using \autoref{prop:unamb-sc}, we can assume in the rest of the proof that $M$ is complete and strongly connected.

We use the fact that one can compute a solution of a possibly singular linear system of equations in~\NCT \cite[Section~5]{BorodinGH82}.
First, compute in~\NCT vectors $\alpha, \beta \in \mathbb{Q}_{>0}^Q$ with $\alpha^T \overline{A} = \alpha^T$ and $\overline{A} \beta = \beta$ and $\alpha^T \beta = 1$.
Using \autoref{lem:U-in-NC} compute in~\NCT a basis of~$U$.
Choose an arbitrary $q \in Q$ and compute $\Mer(q)$ in $\NL \subseteq \NCT$ with a reachability analysis.
Then, solve the linear system from \autoref{thm:pseud-linsyst} to compute in~\NCT a maximal pseudo-column~$y \in \mathbb{Q}^Q$.
Hence, using \autoref{lem:Cas}, we can compute~$\mcw = \alpha^T y$ in~\NCT.
Symmetrically, we can compute~$\mrw$ in~\NCT.
Finally, by \autoref{thm:rank-weight}, we can compute $r = \frac{1}{\mcw \cdot \mrw}$ in~\NCT.
\end{proof}

\begin{exa} \label{ex:snap3}
We continue \autoref{ex:snap1} and \autoref{ex:snap2}.
Since $M$~is a total DFA, it is natural to take $\beta = [Q]$.
Note that $\alpha^T \beta = 1$.
Since every row is of the form $[q]^T$ for some~$q$, we have $\mrw = 1$.
Recall from \autoref{ex:snap2} that $\mcw = \frac12$.
With \autoref{thm:rank-weight} we conclude that $r = \frac{1}{\mcw \cdot \mrw} = 2$, as observed in \autoref{ex:snap1}.
\end{exa}

%% file: sec-time.tex
In this section, we study the time and space complexity of computing the rank of a zero-one matrix monoid and finding a matrix of minimum rank in it. We provide two approaches with the same time complexity. The first one relies on the linear-algebraic tools developed in \autoref{sec-toolbox}. It turns out to have smaller space complexity than the second, combinatorial, approach, but is limited to only computing the rank. This is due to the fact that we never explicitly construct a maximal column in this approach, which is required in order to find a matrix of minimum rank by \autoref{thm:Cesari}. Moreover, it is not known if one can find a maximal column in \NC. In contrast, the combinatorial approach explicitly constructs a matrix of minimum rank step by step. In a way, this is exactly the reason why it requires more space to achieve the same time complexity: in the proof of \autoref{lemma:time-max-col}, we have to precompute a linear number of matrices (one for each step), since we do not know in advance which matrices we will need and since computing them ``on the fly'' would have a higher time complexity.
In the total DFAs case, where transition matrices require only linear amount of space and can be multiplied in linear time, the combinatorial approach becomes much more efficient, and outmatches its linear-algebraic counterpart by a factor of the alphabet size.

The three main results of this section are as follows.

\begin{thm}\label{thm:linalg-time}
    The (real) rank of an $n$-state UFA over an alphabet of size $m$ can be computed in $\OO(mn^4)$ time and $\OO(n^2)$ space. 
\end{thm}

\begin{thm}\label{thm:finding-matrix}
 A matrix of minimum (real) rank in an $n$-state UFA over an alphabet of size~$m$ can be found in  $\OO(mn^4)$ time and $\OO(n^3)$ space. 
\end{thm}

\begin{thm}\label{thm:main-dfas}
    A matrix of minimum (real) rank in an $n$-state total DFA over an alphabet of size~$m$ can be found in $\OO(n^3 + mn^2)$ time and $\OO(n^2)$ space.
\end{thm}

% Fix indentation
\noindent
Until the end of the section, fix a strongly connected complete UFA $\Aa = (Q, \Sigma, \Delta)$. Denote $n = |Q|$, $m = |\Sigma|$. \autoref{subsec:non-sc} shows that strong connectivity can be assumed without loss of generality.

\subsection{Square automaton and square digraph}\label{subsec:square-aut}

We will need the construction of the square automaton of an NFA. The \emph{square automaton} 
$\Aa^{(2)} = (Q^{(2)}, \Sigma, \Delta^{(2)})$ of $\Aa$ is defined as follows. Let $Q^{(2)} = \{(p, q) \mid p, q \in Q\}$, and for $p, q \in Q$ and $a \in \Sigma$, the transitions are defined component-wise, that is, 
\[\Delta^{(2)} = \{ ((p, q), a, (p', q')) \in Q^{(2)} \times \Sigma \times Q^{(2)} \mid (p, a, p'), (q,a,q') \in \Delta\}.\]
Note that the square automaton of a total DFA is also a total DFA.

We call states of the form $(q, q)$ in $\Aa^{(2)}$ \emph{singletons}. Observe that the restriction of~$\Aa^{(2)}$ to singletons is equal to $\Aa$.
We denote by $G^{(2)} = (V^{(2)}, E^{(2)})$ the underlying digraph of~$\Aa^{(2)}$ obtained by forgetting the labels of the transitions. Note that $|E^{(2)}| = \OO(mn^4)$, and there exists an infinite series of complete UFAs over a two-letter alphabet with $|E^{(2)}| = \Theta(n^4)$~\cite[Appendix A]{Kiefer2019Arxiv}. If $\Aa$ is a total DFA, then $|E^{(2)}| = mn^2$.

\subsection{Minimum rank in \texorpdfstring{$\OO(mn^4)$}{O(n4)} time}\label{subsec:linalg-time}

We now perform the steps of the linear-algebraic algorithm described in \autoref{sec-nc2}, but implement them efficiently in terms of time and space complexity.

\begin{lem}\label{lem:mer-time}
    For a state $p$, the set $\Mer(p)$ can be computed in $\OO(mn^4)$ time and $\OO(n^2)$ space.
\end{lem}
\begin{proof}
    Perform a multi-source backwards digraph search starting from all singletons in $G^{(2)}$ and label all states $q \in Q$ such that $(p, q)$ or $(q, p)$ is visited during this search. This search can be performed in time linear in the number of edges of $|E^{(2)}|$, and $|E^{(2)}| = \OO(mn^4)$. 
    Observe that only the vertices of this digraph have to be stored in the memory explicitly, since the edges can be computed on the fly from the input without increasing the time complexity, hence the space complexity of the algorithm is $\OO(n^2)$.
\end{proof}

\begin{lem}\label{lem:pseudcolumn-time}
    A maximal pseudo-column can be found in $\OO(mn^4)$ time and $\OO(n^2)$ space.
\end{lem}
\begin{proof}
To use \autoref{thm:pseud-linsyst} we first need to set up the linear system of equations described there.
The average matrix $\overline{A}$ can be computed in time $\OO(m n^2)$.
The weight vectors $\alpha, \beta \in \mathbb{Q}^Q$ can then be computed in time $\OO(n^3)$ by solving a system of linear equations.
Then, using \autoref{lem:mer-time},
we compute in $\OO(mn^4)$ time the set $\Mer(p)$ for some state $p$.
We also need to compute a basis of the vector space~$U$ defined in \autoref{subsec:pseudocolumns}.
As in the proof of \autoref{lem:U-in-NC}, we compute a basis of $U = U' = \spn{\alpha^T M'(w) \mid w \in \Sigma^+}$, which is the smallest vector space that contains $\alpha^T M(a)$ for all $a \in \Sigma$ and is closed under post-multiplication with $M(a)$ for all $a \in \Sigma$.
This can be done in $\OO(mn^3)$ time and $\OO(n^2)$ space using a worklist algorithm and keeping a basis in echelon form using Gaussian elimination, as described, e.g., in \cite[Section 2]{KieferWeighted20}.
Finally, we solve the system of linear equations from \autoref{thm:pseud-linsyst}, 
which can be done in $\OO(n^3)$ time.
Each step requires at most $\OO(n^2)$ space.
\end{proof}

% Fix indentation
\noindent
By \autoref{lem:Cas}, we thus get that $\mcw$ and $\mrw$ can be computed in  $\OO(mn^4)$ time and $\OO(n^2)$ space, which together with \autoref{thm:rank-weight} proves \autoref{thm:linalg-time}. 
We remark that for total DFAs the proof of \autoref{lem:pseudcolumn-time} gives $\OO(mn^3)$ time for computing the rank. The combinatorial algorithm provided below will improve it to $\OO(n^3 + mn^2)$ (see \autoref{subsec-dfas}), while additionally finding a matrix of minimum rank.

\subsection{Efficiently constructing a maximal column}\label{subsec:max-col} 

We now consider a more general problem of finding a matrix of minimum rank in a zero-one matrix monoid. As mentioned above, by \autoref{thm:Cesari} we have to explicitly construct a maximal column for that, which turns out to be a more difficult task in terms of the time complexity. We will see in the next subsection that we only need one maximal column (together with a word constructing it) to get a matrix of minimum rank.
The goal of this subsection is thus to show how to efficiently compute a short representation of a word constructing a maximal column, since the word itself may be longer than the time complexity we are aiming for. We do so by reusing repeating subwords in such a word, and describing their occurrences with a straight line program.

In \cite[Section 5]{Eppstein1990}, an algorithm for constructing a word of rank one in total DFAs in $\OO(m^3 + mn^2)$ time and $\OO(n^2)$ space was suggested. Our approach for finding a maximal column and a word constructing it follows a similar direction, with a few key differences. Firstly, we observe that it is enough to only compute words merging states with one chosen state~$p$, instead of finding all pairs of mergeable states. This both simplifies the algorithm (in~\cite{Eppstein1990} an additional step is required to decrease the space complexity from $\OO(n^3)$ to $\OO(n^2)$, which we get for free) and allows to present our results in terms of straight line programs, giving a better insight into the regularities present in the constructed word of minimum rank.

We define set straight line programs (set-SLPs), which are just SLPs with multiple initial symbols, and thus encode a set of words instead of one word.
Formally, a \emph{set-SLP} is a tuple $(\VV, \Sigma, R, \Start)$, where $\VV$ and $\Sigma$ are disjoint finite sets of \emph{nonterminal} and \emph{terminal} symbols respectively, $R \colon \VV \to (\VV \cup \Sigma)^*$ is a function defining a derivation rule for each nonterminal symbol, and $\Start \subseteq \VV$ is a set of \emph{initial} symbols. For $v \in \VV$, we write $R(v)$ as $v \to w$ with $w \in (\VV \cup \Sigma)^*$, and we call $v$ and $w$ the left- and right-hand sides of this derivation rule respectively. The \emph{length} of a set-SLP is the total length of the right-hand sides of all the derivation rules.
The semantics of a set-SLP is defined as follows. Given an initial symbol $s \in \Start$, we recursively replace each symbol in $R(s)$ with the right-hand side of its derivation rule until we obtain a word over $\Sigma$, which is called \emph{the word encoded by~$s$}.
We require that each initial symbol produces a unique word over $\Sigma$ as a result of such derivation. Namely, we require that there exists a total linear order $\le$ on the set $\VV$ such that for all $v$ with $v \to w$, $w$ does not contain $v' \in \VV$ with $v' \le v$. The (multi-)set of words encoded by all initial symbols is called \emph{the set of words encoded by} a set-SLP.

\begin{exa}\label{ex:set-slp}
    Consider a set-SLP $(\{w_1, w_3, w_5, u_1, u_2, u_3\}, \{a, b\}, R, \{w_1, w_3, w_5\})$ with 
\[w_1 \to u_1, w_3 \to u_3 u_2, w_5 \to u_2, 
u_1 \to aab, u_2 \to aab, u_3 \to aaba.\]
This set-SLP encodes the (multi-)set $\{aab, aabaaab, aab\}$, and illustrates the reason why we are using set-SLPs: they allow to construct sets of words out of smaller ``pieces'' without having to explicitly repeat these ``pieces'' multiple times (in our example, we are reusing~$u_2$). Note that the set-SLPs that we construct below only encode sets of words whose total length is polynomial in the size of the set-SLPs.
\end{exa}

\begin{lem}
\label{lemma:fast-pair}
    Given a state $p$, a set-SLP of length $\OO(n^2)$ defining a set $\{w_q \mid q \in \Mer(p)\}$, where $w_q$ is a word with $p \in p \cdot w_q$ and $p \in q \cdot w_q$, can be computed in $\OO(mn^4)$ time and $\OO(n^2)$ space.
\end{lem}
\begin{proof}[Proof sketch]
Call vertices $(p, q)$ with $q \in \Mer(p)$ \emph{merging}. The idea is to construct, by a digraph search of $G^{(2)}$, a directed tree $T$ rooted in $(p, p)$  and containing a path from each merging vertex to the root, and then use the joint subpaths of these paths in the tree to obtain a short set-SLP describing these paths. See \autoref{fig:square} for an example.
\end{proof}
\begin{proof}[Proof of \autoref{lemma:fast-pair}]
Call vertices $(p, q)$ with $q \in \Mer(p)$ \emph{merging}. The idea is to construct a directed tree rooted in $(p, p)$ and containing paths from each merging vertex to the root, and then use the joint subpaths of these paths in the forest to obtain a short set-SLP description of these paths.

First, by performing a backwards digraph search in $G^{(2)}$ starting from $(p, p)$, find a subgraph $T$ of $G^{(2)}$ with the following properties (see \autoref{fig:square} for an illustration):
\begin{itemize}
    \item $T$ is a directed tree directed towards its root $(p, p)$;

    \item if $q \in \Mer(q)$, then $T$ contains $(p, q)$ or $(q, p)$;

    \item every leaf of $T$ is a merging vertex, and these are the only leaves in $T$.
\end{itemize}

% Fix indentation
\noindent
Call a directed path $\rho$ in $T$ a \emph{maximal branch} if all its edges belong to $T$, every vertex of~$\rho$ except the last one has outdegree exactly one in $T$, and merging vertices only occur as the first or the last vertex in $\rho$. Intuitively, maximal branches are paths in $T$ between each pair of consecutive branching, leaf, root or merging vertices of~$T$. See \autoref{fig:square} for an example.
Clearly, computing $T$ and then $\rho_1, \ldots, \rho_k$ can be done in time linear in $|E^{(2)}|$.

Let $w_1, \ldots, w_k$ be the words labeling the paths $\rho_1, \ldots, \rho_k$. 
By going backwards from the root of $T$, we can construct a set-SLP for the required words by expressing them as concatenations of $w_1, \ldots, w_k$. This can easily be done in time linear in $|E^{(2)}|$.

It remains to estimate the length of the obtained set-SLP. Since $T$ is a tree, the total length of all words $w_1, \ldots, w_k$ is at most $|V^{(2)}| = n^2$. The number of maximal branches $k$ is~$\OO(n)$. Hence, the length of the obtained set-SLP is~$\OO(n^2)$.
\end{proof}

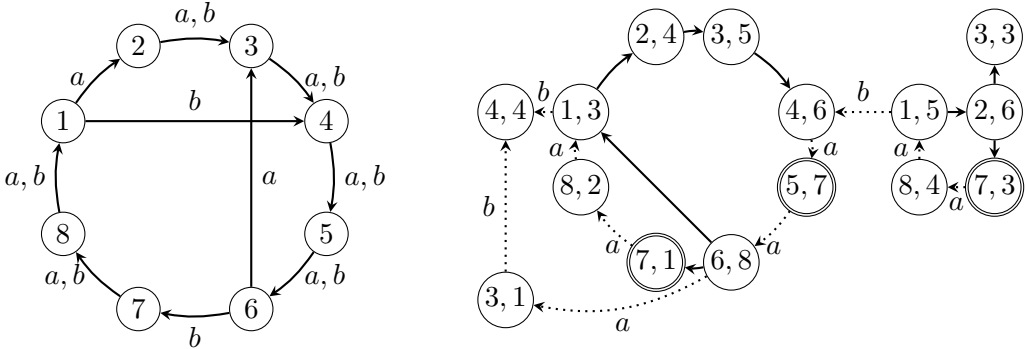
\begin{figure}[h]
\centering
% Center the two drawings as a group
\begin{tabular}{@{}c@{\hspace{1cm}}c@{}}
\begin{tikzpicture} [node distance = 2cm]
%,
\tikzset{every state/.style={inner sep=1pt,minimum size=1.5em}}

\node [state] at (0, 0.5) (1) {$1$};
\node [state] at (1, 1.5) (2) {$2$};
\node [state] at (2.5, 1.5) (3) {$3$};
\node [state] at (3.5, 0.5) (4) {$4$};
\node [state] at (3.5, -1) (5) {$5$};
\node [state] at (2.5, -2) (6) {$6$};
\node [state] at (1, -2) (7) {$7$};
\node [state] at (0, -1) (8) {$8$};

\path [-stealth, thick]
(1) edge [bend left=10] node[left] {$a$} (2)
(2) edge [bend left=10] node[above] {$a, b$} (3)
(3) edge [bend left=10] node[right] {$a, b$} (4)
(4) edge [bend left=10] node[right] {$a, b$} (5)
(5) edge [bend left=10] node[right] {$a, b$} (6)
(6) edge [bend left=10] node[below] {$b$} (7)
(7) edge [bend left=10] node[left] {$a, b$} (8)
(8) edge [bend left=10] node[left] {$a, b$} (1)

(1) edge [] node[above] {$b$} (4)
(6) edge [] node[right] {$a$} (3)

;
\end{tikzpicture}
&
\begin{tikzpicture} [node distance = 2cm]
%,
\tikzset{every state/.style={inner sep=1pt,minimum size=1.5em}}

\node [state] at (-1, 0) (4) {$4, 4$};

\node [state] at (0, 0) (13) {$1, 3$};

\node [state] at (-1, -2.5) (31) {$3, 1$};
\node [state] at (1, 1) (24) {$2, 4$};
\node [state] at (2, 1) (35) {$3, 5$};
\node [state] at (3, 0) (46) {$4, 6$};
\node [state,accepting] at (3, -1) (57) {$5, 7$};
\node [state] at (2, -2) (68) {$6, 8$};
\node [state,accepting] at (1, -2) (17) {$7, 1$};
\node [state] at (0, -1) (28) {$8, 2$};

\node [state] at (4.5, 0) (15) {$1, 5$};
\node [state] at (5.5, 0) (26) {$2, 6$};
\node [state,accepting] at (5.5, -1) (37) {$7, 3$};
\node [state] at (4.5, -1) (48) {$8, 4$};

\node [state] at (5.5, 1) (3) {$3, 3$};

\path [-stealth, thick]
(13) edge [bend left=10] node[above] {} (24)
(24) edge [bend left=10] node[above] {} (35)
(35) edge [bend left=10] node[above] {} (46)

(46) edge [dotted, bend left=10] node[right] {$a$} (57)
(57) edge [bend left=10, dotted] node[below] {$a$} (68)
(68) edge [bend left=10] node[above] {} (17)
(17) edge [bend left=10, dotted] node[below] {$a$} (28)
(28) edge [bend left=10, dotted] node[left] {$a$} (13)

(13) edge [dotted] node[above] {$b$} (4)
(68) edge [] node[above] {} (13)

(68) edge [dotted, bend left=20] node[below] {$a$} (31)
(31) edge [dotted] node[left] {$b$} (4)

(15) edge [] node[above] {} (26)
(26) edge [] node[above] {} (37)
(37) edge [dotted] node[below] {$a$} (48)
(48) edge [dotted] node[left] {$a$} (15)

(15) edge [dotted] node[above] {$b$} (46)

(26) edge [] node[above] {} (3)

;
\end{tikzpicture}
\end{tabular}
\caption{An example of $\Aa$ (left), and a part of the underlying digraph $G^{(2)}$ of its square automaton (right). Merging vertices are doubly circled. The edges of $T$ for $p = 7$ are represented by dotted edges, and these edges are labelled with one of the letters labelling the corresponding transition in $\Aa^{(2)}$. Furthermore, we have $\rho_1 = (7, 1) \to (8, 2) \to (1, 3) \to (4, 4)$, $\rho_2 = (5, 7) \to (6, 8) \to (3, 1) \to (4, 4)$, $\rho_3 = (7, 3) \to (8, 4) \to (1, 5) \to (4, 6) \to (5, 7)$. A set-SLP encoding the labels of these path is presented in \autoref{ex:set-slp}, with~$u_i$ labelling the path $\rho_i$, $i \in \{1, 2, 3\}$.
}\label{fig:square}
\end{figure}

% Fix indentation
\noindent
To construct a maximal column more efficiently, we will use \cite[Corollary 11]{Chistikov2026Arxiv}.

% propC: cite-only heading (avoid double parentheses)
\begin{propC}[\cite{Chistikov2026Arxiv}]\label{prop:mat-mult}
    Given $M_1, M_2 \in \{0, 1\}^{n \times n}$ such that $M_1 M_2 \in \{0, 1\}^{n \times n}$, one can compute $M_1 M_2$ in $\OO(n^2)$ time.
\end{propC}

% Fix indentation
\noindent
Since in the remainder of this section we are dealing only with matrices from a fixed zero-one matrix monoid, \autoref{prop:mat-mult} means that we can always perform matrix multiplication in quadratic time.

\begin{lem}\label{lemma:time-max-col}
    For a given state $p$ of $\Aa$, an SLP of length $\OO(n^2)$ encoding a word $w$ such that $[w \cdot p]$ is a maximal column can be computed in $\OO(mn^4)$ time and $\OO(n^3)$ space.
\end{lem}
\begin{proof}[Proof sketch]
Compute the matrices of the words encoded by the set-SLP from \autoref{lemma:fast-pair}. 
We rely on the property, already used in some form in \cite{Carpi1988}, that if for all $q \in \Mer(p)$, $q \ne p$, the vector $[w \cdot q]$ is zero, then $[w \cdot p]$ is a maximal column. To construct a word with this property, we iteratively concatenate the words $w_q$ depending on nonzero columns in the matrix in each iteration. The number of iterations is bounded by $n$.
\end{proof}

\begin{proof}[Proof of \autoref{lemma:time-max-col}]
By \autoref{lemma:fast-pair}, 
we can compute in $\OO(mn^4)$ time a set-SLP of length $\OO(n^2)$ defining a set $\{w_q \mid q \in \Mer(p)\}$, where $w_q$ is a word with $p \in p \cdot w_q$ and $p \in q \cdot w_q$. For each~$q \in \Mer(p)$, we then compute $M(w_q)$, which in total requires $\OO(n^2)$ matrix multiplications and $\OO(n^3)$ space.

Perform now the following algorithm that iteratively constructs SLPs for words $w_i$ and~$M(w_i)$ for~$i \ge 0$. Let $w_0$ be the empty word. Assume that an SLP for $w_{i - 1}$ and $M(w_{i - 1})$ are already constructed. If there is a state $q \in \Mer(p)$, $q \ne p$, such that $[w_{i - 1} \cdot q] = M(w_{i - 1})[q]$ is not the zero vector, take $w_i = w_{i - 1} w_q$ and compute $M(w_i)$, otherwise stop and output~$w_{i - 1}$. Clearly, each step of the algorithm only requires a constant number of matrix multiplications.

The algorithm clearly terminates in at most $|\Mer(p)|$ steps, since each iteration decrements by at least one the number of states $q \in \Mer(p)$ such that $[w_i \cdot q]$ is not the zero vector. It is also easy to see that the constructed SLP for $w$ has length $\OO(n^2)$.

Assume now that $[w \cdot p]$ is not a maximal column. Let $C \subseteq Q$ be such that $[C] \ge [w \cdot p]$ is a maximal column. By definition of $\Mer(p)$, $C \subseteq \Mer(p)$, hence for every $q \in C$, $q \ne p$, the vector $[w \cdot q]$ is zero. Since $[C] \ge [w \cdot p]$, $p \in C$. Thus we get that $M(w)[C] = \sum_{q \in C} [w \cdot q] = [w \cdot p]$.
By \autoref{lem:max-column-stable}, $M(w)[C]$ is a maximal column, hence $[w \cdot q]$ must also be a maximal column. 
\end{proof}

\subsection{Finding a matrix of minimum rank}\label{subsec-finding-mat}

We now use the results of the previous section to construct a matrix of minimum rank. The key idea is as follows: since the set of maximal columns is stable under left multiplication by matrices from the monoid, we can iteratively make each column of the matrix maximal or zero by, intuitively, applying the same word (together with a short ``reachability'' word) to a state in each column. This simple observation significantly decreases the time complexity of our algorithm compared to the naive implementation of the algorithm from \cite{Carpi1988} constructing a word of minimum rank. Indeed, in the approach of \cite{Carpi1988}, the word is constructed letter by letter, and requires to know the result of applying the last letter at every step. Since the constructed word has length $\OO(rn^3) = \OO(n^4)$, where $n$ is the number of states and $r$ is the rank, this results in $\OO(n^{4+\omega})$ time complexity. By efficiently constructing only one maximal column (as described in the previous section) and reusing it for the whole set of states (as described in this section), we decrease the time complexity to $\OO(n^{2+\omega})$.

\begin{prop}\label{prop:time-minrank-mat}
    An SLP of length $\OO(n^2)$ encoding a word $w$ of minimum rank can be computed in $\OO(mn^4)$ time and $\OO(n^3)$ space.
\end{prop}
\begin{proof}[Proof sketch]
 Compute the matrix of the word $w$ encoded by the SLP from \autoref{lemma:time-max-col}. Iteratively, for each $q \in Q$, concatenate $w$ with a word of length at most $n$ mapping $p$ to a state corresponding to a nonzero element of the row in the current iteration. Denote by $w_n$ the resulting word, which has the property that all nonzero columns of $M(w_n)$ are maximal. Symmetrically compute $w'_n$ for rows. Then by \autoref{lem:all-max-column} the word $w_{n}w'_{n}w_{n}w'_{n}$ matrix has minimum rank.
\end{proof}
\begin{proof}[Proof of \autoref{prop:time-minrank-mat}]
By \autoref{lemma:time-max-col},
for a given state $p$ of $\Aa$, we can compute in $\OO(mn^4)$ time an SLP of length $\OO(n^2)$ defining a word $w$ such that $[w \cdot p]$ is a maximal column. Compute $M(w)$, which can be done in  $\OO(n^4)$ time.

Let $Q = \{q_1, \ldots, q_n\}$. Define $v_0 = \epsilon$. Clearly, $M(v_0)$ is the identity matrix. 
For each $1 \le i \le n$, perform the following algorithm. If $[v_{i - 1} \cdot q_i]$ is the zero vector, take $v_i = v_{i - 1}$ and go the the next step. Otherwise, find $q$ such that $q \in v_{i - 1} \cdot q_i$, and 
find a word $u_{p \to q}$ of length at most $n - 1$ mapping $p$ to $q$, and compute $M(u_{p \to q})$, which can be done by~$\OO(n)$ matrix multiplications. Take $w_i = w u_{p \to q} w_{i - 1}$, and compute $M(w_i)$, which requires a constant number of matrix multiplications. Go to the next step of the algorithm.

Observe that after the $i$th step of the algorithm, $[v_j \cdot q_j]$ is a maximal column. Indeed, $[v_i \cdot q_i] = [w \cdot p]$, and for all $j < i$, $[v_j \cdot q_j]$ is a result of left multiplication of a maximal column by a matrix from the monoid, which is a maximal column by \autoref{lem:max-column-stable}.
Moreover, the obtained SLP of $w_{n}$ has length $\OO(n^2)$.

It remains to construct symmetrically $w'_{n}$ for rows instead of columns. Then each column and each row of $M(w_{n}w'_{n})$ is either maximal or zero, and by \autoref{lem:all-max-column} the word $w_{n}w'_{n}w_{n}w'_{n}$ matrix has minimum rank.
\end{proof}

% Fix indentation
\noindent
Given an SLP of length $\OO(n^2)$ encoding a word $w$, we can compute the matrix of $w$ by computing the matrices of words occurring in the derivation of $w$ from bottom to top in time $\OO(n^4)$.
Thus we prove \autoref{thm:finding-matrix}. We also get the following result, which can be seen as a proof of a very weak version of the \v{C}ern\'{y} conjecture generalised from rank one words in total DFAs to minimum rank words in complete UFAs.

\begin{thm}
    For every $n$-state complete UFA, there exists an SLP of length $\OO(n^2)$ encoding a word of minimum rank.
\end{thm}

% Fix indentation
\noindent
We remark that the length of the word encoded by the constructed SLP asymptotically matches the best known upper bound  for words of minimum rank: $\OO(n^4)$ for complete UFAs~\cite{Carpi1988} and $\OO(n^3)$ for total DFAs~\cite{Klyachko1987}. In particular, one can efficiently compute words of minimum rank within these bounds.

\subsection{Total DFAs}\label{subsec-dfas}

For total DFAs, we follow the same algorithms as in the proof of \autoref{thm:finding-matrix}, but exploit the fact that elementary matrix operations can be performed more efficiently. Namely, if $\Aa$ is a total DFA, then each word defines a transformation on $Q$. By storing matrices of words as transformations, we get that matrix multiplication can be performed in $\OO(n)$ time, and each matrix requires $\OO(n)$ space. Moreover, we have $|E^{(2)}| = mn^2$. By taking these improvements into account, we get the proof of \autoref{thm:main-dfas}.

%% file: sec-algebraic-criterion.tex
The rank of~$M$ can be viewed as a combinatorial property, in that matrix multiplication is not commutative, and even the rank of a matrix product can depend on the order of the multiplied matrices.
Extending the results of \autoref{sec-toolbox} but aiming at a more structural result, we address a more general question in this section: is there a specific vector space, perhaps a joint invariant subspace of the generating matrices, whose dimension tells us something about the rank of $M$? 

In total DFAs, every row $[q]^T$ is maximal, so $\spn{\MRow} = \mathbb{R}^Q$. In \cite[Criterion~1]{DBLP:journals/isci/BerlinkovS16}, the following result was proved. A different proof was independently and concurrently provided in \cite[\S8]{Voynov2015Compact}, see also \cite{Protasov2021}\footnote{ \cite[Theorem 2]{Protasov2021} states that if the rank of a total DFA is greater than one, then there exists a non-trivial joint invariant subspace of its generating matrices. However, such a subspace exists for all total DFAs regardless of their rank~\cite[Corollary 4]{Protasov2017}. 
Hence, to characterise total DFAs of rank greater than one, we have to consider the existence of some specific joint invariant linear subspace, for example the one in the statement of \autoref{thm:alg-criterion-dfas}.}.

\begin{thm}[Algebraic synchronization criterion for total DFAs]\label{thm:alg-criterion-dfas}
    If $M$ is a total DFA, we have 
    $\spn{\alpha^T M(w) \mid w \in \Sigma^*} = \mathbb{R}^Q$
    if and only if $r=1$.
\end{thm}

% Fix indentation
\noindent
It is thus reasonable to ask if this statement can be generalised to the case where $r > 1$ and $M$ is an arbitrary unambiguous monoid morphism. The theorem below provides such a generalisation, using the results obtained above. In particular, it implies that for total DFAs $\dim \spn{\alpha^T M(w) \mid w \in \Sigma^*} \le n - r + 1$. To be consistent with the previous results, we formulate the statements for columns, but an analogous symmetric version for rows, involving $\spn{\alpha^T M(w) \mid w \in \Sigma^*}$ and $\spn{\MRow}$ instead of, respectively, $V$ and $\spn{\MCol}$ also holds true. 

\begin{thm}\label{thm:algebraic}
Define $V := \spn{M(w) \beta \mid w \in \Sigma^*}$.
We have:
\begin{enumerate}
\item[(a)] $V \subseteq \spn{\MCol}$.
\item[(b)] $\dim V + r-1 \le \dim \spn{\MCol}$.
\item[(c)] $V = \spn{\MCol}$ if and only if $r=1$.
\end{enumerate}
\end{thm}

\begin{proof}
Towards item~(a), it suffices to show that $\spn{\MCol}^\bot \subseteq V^\bot$.
Let $x^\top \in \spn{\MCol}^\bot$.
For every word $u \in \Sigma^*$ of minimum rank, by \autoref{thm:Cesari}, each column of~$M(u)$ is in $\spn{\MCol}$.
Thus, $x^\top M(u) \beta = 0$ holds for all $u \in \Sigma^*$ of minimum rank.
By \autoref{prop:compact}, it follows that $x^\top M(w) \beta = 0$ holds for all $w \in \Sigma^*$, i.e., $x^\top \in V^\bot$.

Towards item~(b), let $u \in \Sigma^*$ be of minimum rank.
By \autoref{thm:Cesari}, there are $r$ maximal columns $[C_1], \ldots, [C_r]$ and $r$ maximal rows $[R_1]^T, \ldots, [R_r]^T$ such that $M(u) = \sum_{i=1}^r [C_i] [R_i]^T$.
We conclude from item~(a) that $\spn{V, [C_1], \ldots, [C_r]} \subseteq \spn{\MCol}$.
Hence, it suffices to show that $\dim V + r-1 = \dim \spn{V, [C_1], \ldots, [C_{r-1}]}$.

We have
\[
M(u u) \ = \ \left( \sum_{i=1}^r [C_i] [R_i]^T \right) \left( \sum_{j=1}^r [C_j] [R_j]^T \right)
\ = \ \sum_{1 \le i,j \le r} [C_i] [R_i]^T [C_j] [R_j]^T\,.
\]
By unambiguousness, $[R_i]^T [C_j] \le 1$ for all $1 \le i,j \le r$.
For each~$i$ there is at most one~$j$ with $[R_i]^T [C_j] = 1$ as otherwise $[R_j]^T + [R_{j'}]^T$ for some $j' \ne j$ would appear in~$M(u u)$, contradicting the maximality of~$[R_j]^T$.
On the other hand, for each~$i$ there is at least one~$j$ with $[R_i]^T [C_j] = 1$ as otherwise $M(u u)$ has only $[C_{i'}]$ with $i' \ne i$ as nonzero columns, contradicting the fact that $M(u u)$ has rank~$r$.
Thus, for each~$i$ there is exactly one~$j$ with $[R_i]^T [C_j] = 1$.
With a symmetric argument, for each~$j$ there is exactly one~$i$ with $[R_i]^T [C_j] = 1$.
It follows that there is permutation $\pi: \{1, \ldots, r\} \to \{1, \ldots, r\}$ such that $[R_i]^T [C_j] = 1$ if and only if $i = \pi(j)$.

To show that $\dim V + r-1 = \dim \spn{V, [C_1], \ldots, [C_{r-1}]}$, it suffices to show for each $j \in \{1, \ldots, r-1\}$ that $[C_j] \not\in \spn{V, [C_1], \ldots, [C_{j-1}]}$.
To this end, let $1 \le j < r$ and define $x^T := [R_{\pi(j)}]^T - [R_{\pi(r)}]^T$.
Since $[R_{\pi(j)}]^T$ and $[R_{\pi(r)}]^T$ are maximal rows, for all $w \in \Sigma^*$ also $[R_{\pi(j)}]^T M(w)$ and $[R_{\pi(r)}]^T M(w)$ are maximal rows, implying that $x^T M(w) \beta = \mrw - \mrw = 0$; i.e., $x^T$ is in the orthogonal complement of~$V$.
Further, for all $i \in \{1, \ldots, j-1\}$ we have $x^T [C_i] = [R_{\pi(j)}]^T [C_i] - [R_{\pi(r)}]^T [C_i] = 0-0 = 0$; i.e., $x^T$~is orthogonal also to~$[C_i]$.
But $x^T$~is not orthogonal to~$[C_j]$, as $x^T [C_j] = [R_{\pi(j)}]^T [C_j] - [R_{\pi(r)}]^T [C_j] = 1 - 0 = 1$.
This completes the proof of item~(b).

Towards item~(c), if $r>1$ we have $V \subsetneq \spn{\MCol}$ by item~(b).
Let $r=1$.
By item~(a) we have $V \subseteq \spn{\MCol}$.
Towards the opposite inclusion, let $y \in \MCol$.
Since $r=1$, there is $w \in \Sigma^*$ such that $y$ is the only nonzero column of~$M(w)$.
Thus, $y$~is a multiple of $M(w) \beta$.
Hence, $y \in V$.
\end{proof}

\begin{rem}
The inequality in \autoref{thm:algebraic}~(b) and its row version can be both strict if $r>1$, even in total DFAs. In \autoref{ex:snap1}, $V$ is spanned by $\begin{pmatrix} 1 & 1 & 1 & 1\end{pmatrix}^T$, but 
\[\dim \spn{\MCol} = 3 > 1 + 2 - 1 = \dim V + r-1.\]
For the row version, $\spn{\alpha^T M(w) \mid w \in \Sigma^*}$ is spanned by 
$\begin{pmatrix}1 & 0 & 1 & 0\end{pmatrix}$ and
$\begin{pmatrix}0 & 1 & 0 & 1\end{pmatrix}$, but 
\[\dim \spn{\MRow} = 4 > 2 + 2 - 1 = \dim \spn{\alpha^T M(w) \mid w \in \Sigma^*} + r - 1.\]
\end{rem}

%% file: sec-conclusions.tex
We list a few open questions that follow directly from our work.

\begin{itemize}
    \item In \cite{Eppstein1990}, it was asked if a word of rank one for a total DFA can be found in \NC. Similarly, can a matrix of minimum rank for a total DFA be computed in \NC?

    \item Given an unambiguous morphism $M \colon \Sigma \to \{0, 1\}^{Q \times Q}$ and a vector $\alpha \in \mathbb{Q}_{>0}^Q$, can a basis of  $\spn{\alpha^T M(w) \mid w \in \Sigma^*}$ be computed faster than in $\OO(|Q|^3)$ time? This would improve algorithms for several fundamental problems for weighted automata~\cite{KieferWeighted20}. Similarly, for total DFAs, computing a basis of $U$ from \autoref{subsec:pseudocolumns} in subcubic time (see the proof of \autoref{lem:pseudcolumn-time}) would allow to compute the rank of a total DFA faster than in cubic time.

    \item Can one decide if a total DFA has rank one in strongly subquadratic time (in the number of states)? This seems to be a major open problem in the area of synchronising automata.
    
    \item The bottleneck in the time complexity in \autoref{thm:linalg-time} is the very first step, computing $\Mer(q)$ via digraph search in the square digraph of $\Aa$. The number of edges of this digraph can be quadratic in the number of its vertices~\cite[Appendix A]{Kiefer2019Arxiv}, hence of order~$|Q|^4$. Can $\Mer(q)$ be computed faster than in time $\OO(|Q|^4)$? Very little seems to be known about general properties of square automata of DFAs or UFAs.

    \item Finally, a natural continuation of this work is to consider the minimum nonzero rank of zero-one matrix monoids. It is equal to the rank of each nonzero matrix in the 0-minimal ideal of the monoid. It is not known how to compute it in \NC even for total DFAs. The main motivation once again comes from the degree of variable-length codes, see \cite[Chapter 9]{Berstel2010} for more details. 
\end{itemize}